\documentclass[11pt]{article}

\usepackage[a4paper,margin=1in]{geometry}

\usepackage[T1]{fontenc}
\usepackage[utf8]{inputenc}
\usepackage{lmodern}
\usepackage{microtype}

\usepackage{amsmath,amssymb,amsfonts,amsthm,mathtools}
\numberwithin{equation}{section}

\usepackage{graphicx}
\usepackage{subcaption}
\usepackage{booktabs}
\usepackage{threeparttable}

\usepackage{algorithm}
\usepackage{algorithmic}

\usepackage[numbers]{natbib}
\usepackage[hidelinks]{hyperref}

\usepackage{accents}

\theoremstyle{plain}
\newtheorem{theorem}{Theorem}[section]
\newtheorem{lemma}[theorem]{Lemma}

\theoremstyle{remark}

\usepackage[font=small,labelfont=bf]{caption}

\title{\Large\bfseries
Improving the adjusted Benjamini-Hochberg method using $e$-values in knockoff-assisted variable selection
}

\author{
Aniket Biswas\thanks{Corresponding author: biswasaniket44@gmail.com} \\
\small Department of Mathematics, Indian Institute of Technology Kharagpur, India
\and
Aaditya Ramdas \\
\small Department of Statistics and Data Science, Carnegie Mellon University, USA
}

\date{}

\begin{document}
\maketitle

\begin{abstract}
\noindent Considering the knockoff-based multiple testing framework of \citet{barber2015controlling}, we revisit the method of \citet{sarkar2022adjusting} and identify it as a specific case of an un-normalized $e$-value weighted Benjamini-Hochberg procedure. Building on this insight, we extend the method to use bounded $p$-to-$e$ calibrators that enable more refined and flexible weight assignments. Our approach generalizes the method of \citet{sarkar2022adjusting}, which emerges as a special case corresponding to an extreme calibrator. Within this framework, we propose three procedures: an $e$-value weighted Benjamini-Hochberg method, its adaptive extension using an estimate of the proportion of true null hypotheses, and an adaptive weighted Benjamini-Hochberg method. We establish control of the false discovery rate (FDR) for the proposed methods. While we do not formally prove that the proposed methods outperform those of \citet{barber2015controlling} and \citet{sarkar2022adjusting}, simulation studies and real-data analysis demonstrate large and consistent improvement over the latter in all cases, and better performance than the knockoff method in scenarios with low target FDR, a small number of signals, and weak signal strength. Simulation studies and a real-data application in HIV-1 drug resistance analysis demonstrate strong finite sample FDR control and exhibit improved, or at least competitive, power relative to the aforementioned methods.  
\end{abstract}

\paragraph{Keywords:}
Multiple testing; False discovery rate; Bonferroni--Benjamini--Hochberg method; $E$-value; Linear regression.

\paragraph{MSC (2020):}
62F99; 62J99; 62P10.

\section{Introduction}

Let us consider the classical linear regression model with the response $Y\in \mathbb{R}^n$ and the full column-rank non-stochastic design matrix $X$ of dimension $n\times m$
\begin{equation}\label{eq:regmodel}
Y=X\beta+\varepsilon,
\end{equation}
where $m\leq n$, $\beta=(\beta_1, \beta_2, ..., \beta_m)\in \mathbb{R}^m$ is the unknown vector of regression coefficients and the error vector $\varepsilon\sim N_n(0,\sigma^2I_n)$ with unknown $\sigma^2>0$. For the model (\ref{eq:regmodel}), the ordinary least square estimate of $\beta$, $\hat{\beta}=\Sigma^{-1}X^\top Y\sim N_m(\beta,\sigma^2 \Sigma)$, $\Sigma=X^\top X$, is the best linear unbiased estimate of $\beta$. In many modern-day applications, a large number of independent variables are available for modelling the response. However, there are typically just a few relevant independent variables worth consideration in the model. This problem of variable selection for (\ref{eq:regmodel}) is formed as a multiple testing problem
\begin{equation}\label{eq:testbeta}
    H_j:\beta_j=0\qquad\textrm{vs.}\qquad K_j:\beta_j\neq 0,\quad j=1,2,...,m.
\end{equation}
The independent variables corresponding to the false $H_j$'s in (\ref{eq:testbeta}) are regarded as important variables and retained in the model. The exclusion of the other independent variables make the model (\ref{eq:regmodel}) parsimonious and easily interpretable. For the $j$-th test in (\ref{eq:testbeta}), marginal $t$-test based on $\hat{\beta}_j$ is optimal \citep{lehmann2005testing}, for each $j=1,2,...,m$. However, the problem of variable selection demands control of an overall measure of Type-I error for the set of tests in (\ref{eq:testbeta}) \citep{barber2015controlling}. The false discovery rate (FDR) \citep{benjamini1995controlling} is considered an appropriate choice for the overall measure of Type-I error in the literature. However, the classical Benjamini-Hochberg (BH) method of \citet{benjamini1995controlling} is not appropriate here due to the arbitrary dependence structure among the test-statistics based on $\hat{\beta}$ \citep{sarkar2022adjusting}. 

Considering $n\geq 2m$, \citet{barber2015controlling} constructed a knockoff design matrix $\tilde{X}$ such that $\tilde{X}^\top \tilde{X}=\Sigma$ and $X^\top\tilde{X}=2\Sigma-D$, where $D$ is a diagonal matrix satisfying the condition that $\Sigma-D$ is positive definite. In their seminal work, \citet{barber2015controlling} detail the construction of $D$, $\tilde{X}$ and propose an innovative FDR controlling method. We refer to this $p$-value free method as knockoff method in this work. \citet{sarkar2022adjusting} use $X$, $\tilde{X}$ and $D$ to get back to the classical $p$-value based approach of FDR control for the problem of variable selection. The authors propose two independent estimators of $\beta$:
\begin{eqnarray*}
    & &\hat{\beta}^{(1)}=(2\Sigma-D)^{-1} (X+\tilde{X})^\top Y\sim N_m(\beta,2\sigma^2(2\Sigma-D)^{-1}),\\
    & &\hat{\beta}^{(2)}=D^{-1} (X-\tilde{X})^\top Y\sim N_m(\beta,2\sigma^2 D^{-1}).
\end{eqnarray*}
From these two sets of estimators, the following sets of statistics can be formulated to test $H_j$ for each $j=1,2,...,m$:
\begin{eqnarray*}
    & & T^{(1)}=(T^{(1)}_1,T^{(1)}_2, ..., T^{(1)}_m)=\frac{1}{\hat{\sigma}\sqrt{2}}[\textrm{diag}\{(2\Sigma-D)^{-1}\}]^{-\frac{1}{2}} \hat{\beta}^{(1)},\\
    & & T^{(2)}=(T^{(2)}_1,T^{(2)}_2, ..., T^{(2)}_m)=\frac{1}{\hat{\sigma}\sqrt{2}} D^{\frac{1}{2}}\hat{\beta}^{(2)}.
\end{eqnarray*}
Here we take $\hat{\sigma}$ to be the root mean square of the residuals obtained from the linear regression of the response vector $Y$ with the original design matrix $X$. The degrees of freedom used here is $\nu=n-2m$ when $n>2m$ and $\nu=n-m$ when $m<n\leq 2m$. The two set of $p$-values, corresponding to the testing problems in (\ref{eq:testbeta}), are obtained from $T^{(1)}$ and $T^{(2)}$ as
\begin{eqnarray*}
    P^{(1)}_j=2\times \Pr(\mathcal{T}_\nu> |T^{(1)}_j|)\quad \textrm{and}\quad P^{(2)}_j=2\times \Pr(\mathcal{T}_\nu> |T^{(2)}_j|),\quad j=1,2,...,m,
\end{eqnarray*}
where $\mathcal{T}_\nu$ denotes the $t$-distributed random variable with degrees of freedom $\nu$. \citet{sarkar2022adjusting} used the first set of $p$-values $P^{(1)}=(P^{(1)}_1, P^{(1)}_2, ..., P^{(1)}_m)$ in a Bonferroni-type method with a common threshold and then used the second set of $p$-values $P^{(2)}=(P^{(2)}_1, P^{(2)}_2, ..., P^{(2)}_m)$ in combination with the first in the BH method to get control over the FDR. We call this method the Bonferroni-Benjamini-Hochberg (Bon-BH) method. It is important to note that if $\sigma$ is known, or conditional on the estimate $\hat{\sigma}$, the components of $P^{(2)}$ are independent, and they are also independent of the components of $P^{(1)}$. This assumption is crucial for establishing the asymptotic control of the FDR by the adaptive FDR-controlling methods we propose in Section~3. For finite-sample FDR control by a non-adaptive method proposed in Section~3, the exact dependence structure has been explicitly utilized in the proof. The Bon-BH method achieves higher power than the knockoff filter under weaker signals and low target FDR levels, while maintaining comparable performance in other scenarios. Thus, it offers a competitive alternative under a range of conditions.

The $e$-value is a recently developed inferential tool that unifies concepts such as betting scores, likelihood ratios, Bayes factors, and stopped super-martingales \citep{shafer2021testing, vovk2021evalues, grunwald2024safe, grunwald2024beyond, howard2021time}.
A $p$-value can be converted to an $e$-value using appropriate calibrators \citep{vovk2021evalues}.
Lists of $e$-values can serve as the input of multiple testing procedures \citep{wang2022false} or they can also be incorporated to $p$-value based multiple testing methods as un-normalized weights \citep{ignatiadis2024evalues}. In the present work, the use of $e$-values as weights in the $p$-value based BH method ($ep$-BH method) is of particular interest. When the $e$-values are bounded, they can be standardized to form a set of weights summing to $m$, which can then be used as $p$-value weights in the weighted BH procedure \citep{genovese2006false} or its adaptive variants \citep{Habiger2017adaptive, ramdas2019unified, biswas2023new}.

The main idea of this paper is to reinterpret two-stage multiple testing procedures through the framework of $p$-to-$e$ calibration. From this perspective, the method of \citet{sarkar2022adjusting} arises from a specific and extreme choice of a bounded calibrator applied to the first-stage $p$-values. This insight allows their procedure to be viewed as a special case of our general approach. By replacing this extreme calibration with a more regular bounded calibrator, we obtain methods that maintain theoretical validity and exhibit improved empirical behavior. This viewpoint motivates the proposed methodology and is formalized in Section~4.

\noindent\textbf{Contributions of the paper.}
The main contributions of this work are summarized as follows:
\begin{itemize}
    \item We propose a unified knockoff-assisted multiple testing framework based on $p$-to-$e$ calibration, which enables the construction of valid two-stage variable selection procedures.
    \item We introduce three knockoff-assisted methods with finite-sample or asymptotic FDR guarantees, including a true null proportion–adaptive and a weighted variant.
    \item We show that the Bon-BH and adaptive Bon-BH procedures of \citet{sarkar2022adjusting} arise as special cases of our framework corresponding to a specific bounded calibrator, and we propose a continuous bounded alternative that exhibits favorable empirical performance.
\end{itemize}

The present work not only extends the new line of research combining the seminal ideas of knockoff and $p$-value based multiple testing methods, it also incorporates the recently introduced idea of $e$-values into it. 

The rest of the article is organized as follows. Section~2 briefly reviews some of the existing multiple testing procedures, knockoff-based and knockoff-adjusted variable selection methods. Section~3 presents the proposed methods and establishes their FDR control properties. Section~4 places the methods of \citet{sarkar2022adjusting} within our framework and discusses the role of bounded calibrators. Section~5 presents simulation results on FDR control and power. Section~6 illustrates the methods using a real HIV-1 drug resistance dataset.

\section{Review of related methods}

In this section, we begin with a brief overview of foundational concepts and procedures in multiple testing, including the FDR, the BH method, its adaptive version using Storey’s estimate of the proportion of true null hypotheses, and the weighted BH method. These preliminaries provide the necessary background for discussion on the three general multiple testing procedures that are closely related to our proposed methods: the \(ep\)-BH method, its adaptive variant, and the adaptive weighted BH method.

We then review two related variable selection procedures that also aim to control the FDR: the knockoff method, which is free of \(p\)-values, and the knockoff-assisted Bon-BH method along with its adaptive extension. While the first group comprises general procedures applicable across multiple testing contexts, the second group is specifically designed for variable selection framed through testing of hypotheses.

In this work, we adapt the general procedures to variable selection problems, and our results demonstrate that the resulting methods are not only competitive but also offer promising directions for future development.

\subsection{General procedures}

A $p$-value $P$ corresponding to a null hypothesis $H$ is a $[0,1]$-valued random variable such that $\Pr(P \leq t) \leq t$ for all $t \in [0,1]$, under every distribution in $H$. Similarly, an $e$-value $S$ for a null hypothesis $H$ is a $[0, \infty]$-valued random variable such that $E(S) \leq 1$ under every distribution in $H$. In the case of a simple null hypothesis—that is, when $H$ consists of a single distribution—these conditions can be written as $\Pr_H(P \leq t) \leq t$ for $p$-values and $\mathbb{E}_H(S) \leq 1$ for $e$-values. For composite null hypotheses, where $H$ is a set of distributions, the validity conditions must hold uniformly over all distributions in $H$. Opposite to the case of a $p$-value, a large $e$-value shows evidence against the null hypothesis $H$. 

Consider a general multiple testing problem involving \( m \) null hypotheses:
\(
H_1, H_2, \ldots, H_m.
\)
For each hypothesis \( H_j \), \( j = 1, 2, \ldots, m \), suppose we have access to both a \( p \)-value \( P_j \in [0,1] \) and an \( e \)-value \( S_j \in [0,\infty) \). The collections \( \{P_1, P_2, \ldots, P_m\} \) and \( \{S_1, S_2, \ldots, S_m\} \) summarize the evidence against the null hypotheses in the traditional \( p \)-value and the \( e \)-value testing frameworks, respectively.  

Let \( R \) denote the number of rejections made by a multiple testing procedure, and \( V \) the (unobserved) number of false rejections, i.e., true nulls that are rejected. The false discovery proportion is defined as
\[
\mathrm{FDP} = \frac{V}{R \vee 1} = 
\begin{cases}
V/R, & \text{if } R \geq 1, \\
0, & \text{if } R = 0.
\end{cases}
\] 
Since the proportion of false discovery is unobservable, procedures typically aim to control its expectation, the FDR.

Let \( P_{(1)} \leq P_{(2)} \leq \cdots \leq P_{(m)} \) denote the ordered \( p \)-values, and let \( H_{(j)} \) be the null hypothesis corresponding to \( P_{(j)} \), for each \( j \in \{1, 2, \ldots, m\} \). 

The Benjamini–Hochberg procedure, which controls the FDR at level \( \alpha \in (0,1) \), proceeds as follows: identify
\[
j_0 = \max \left\{ j \in \{1,2,\ldots,m\} : P_{(j)} \leq \frac{\alpha j}{m} \right\},
\]
and reject all hypotheses \( H_{(j)} \) for \( j \leq j_0 \).

Let \(\mathcal{N}\) denote the set of true null hypotheses, \(m_0 = \text{Card}(\mathcal{N})\) be the number of true nulls, and \(\pi_0 = m_0/m\) denote the proportion of true null hypotheses. For a suitably chosen threshold \(\lambda \in (0,1)\), suppose that \(P(P_j > \lambda) = 0\) for each \(j \notin \mathcal{N}\). Then, assuming that \(P_j\) is uniformly distributed for each $j\in\mathcal{N}$, the estimator for \(\pi_0\) proposed by \citet{storey2004strong} is given by
\begin{equation*}
    \hat{\pi}_0 = \frac{1 + \sum_{j=1}^m \mathbb{I}(P_j > \lambda)}{m(1 - \lambda)}.
\end{equation*}
Adaptive BH procedures identify $j_0=\max\{j\in \{1,2,\ldots,m\}:P_{(j)}\leq (\alpha j)/(\hat{\pi}_0 m)\}$ and rejects $\{H_{(j)}:j\leq j_0\}$. \citet{benjamini2006adaptive} proved that the adaptive BH procedure controls the FDR at $\alpha$. 

When the input to the BH method is taken as the transformed set \(\{P_1/S_1, P_2/S_2, \ldots,\) \(P_m/S_m\}\) instead of the usual \(\{P_1, P_2, \ldots, P_m\}\), the resulting procedure is known as the \(ep\)-BH method of \citet{ignatiadis2024evalues}. The adaptive $ep$-BH method, referred to as the \(ep\)-Storey method by \citet{ignatiadis2024evalues}, uses the modified input
\[
\left\{ \frac{\hat{\pi}_0 P_1}{\mathbb{I}(P_1 \leq \lambda) S_1}, \frac{\hat{\pi}_0 P_2}{\mathbb{I}(P_2 \leq \lambda) S_2}, \ldots, \frac{\hat{\pi}_0 P_m}{\mathbb{I}(P_m \leq \lambda) S_m} \right\}
\]
in place of the original \(p\)-values in the BH method. \citet{ignatiadis2024evalues} showed that the ep-BH method controls the false discovery rate at level \(\alpha\) when the null \(p\)-values are positively regression dependent, and that the ep-Storey method achieves the same guarantee under independence of the null \(p\)-values. In both cases, the result holds provided that the null \(p\)-values are independent of the corresponding \(e\)-values.

Let \(\{W_1, W_2, \ldots, W_m\}\) be a set of prior \(p\)-value weights such that \(\sum_{j=1}^m W_j = m\). Ideally, \(W_j\) should be large when \(H_j\) is false and small otherwise, but their specification remains a guess. When the input to the BH method is taken as the transformed set \(\{P_1/W_1, P_2/W_2, \ldots, P_m/W_m\}\) instead of the usual \(\{P_1, P_2, \ldots, P_m\}\), the resulting procedure is known as the weighted BH method of \citet{genovese2006false}.

Unlike the $ep$-BH method, the weighted version requires that the weights sum to \(m\), while no such restriction exists for \(e\)-values. When using e-values, the weights can be constructed as $W_j = S_j$ or as \(W_j = mS_j/\sum_{k=1}^m S_k\), for \(j = 1, 2, \ldots, m\).

If $E(W_j)=\mu_0\leq 1$ for each $j\in\mathcal{N}$, the weighted BH method controls the FDR at \(\delta_0\alpha\), where \(\delta_0 = \pi_0\mu_0 \leq 1\). An adaptive weighted Benjamini-Hochberg procedure, in which the parameter $\delta_0$ is estimated through data-driven weights, has been studied in earlier works; see \citet{ramdas2019unified}, and \citet{biswas2023new}. In this approach, \(\delta_0\) is estimated by
\begin{equation*}
    \hat{\delta}_0 = \frac{\max\{W_1, W_2, \ldots, W_m\} + \sum_{j=1}^m W_j\mathbb{I}(P_j > \lambda)}{m(1 - \lambda)},
\end{equation*}
and the BH procedure is then applied to the adjusted set
\(\{\hat{\delta}_0 P_1/W_1, \hat{\delta}_0 P_2/W_2, \ldots, \hat{\delta}_0 P_m/W_m\}\). This controls the FDR at level \(\alpha\) under the assumption that the set of \(p\)-values \(\{P_j : j \in \mathcal{N}\}\) are independent, and it is independent of the corresponding weights \(\{W_j : j \in \mathcal{N}\}\).

\subsection{Knockoff-filter and knockoff-assisted procedures}

We consider the regression model in (\ref{eq:regmodel}) under the assumption that $n \geq 2m$, and fix the FDR level at $\alpha \in (0,1)$. 

\textbf{Knockoff method:} We briefly outline the knockoff method of \citet{barber2015controlling}. The construction and selection steps are as follows:

\begin{itemize}
    \item \textit{Knockoff construction}: As introduced in Section~1, the knockoff matrix \(\tilde{X}\) is constructed via
    \[
    \tilde{X} = X \Sigma^{-1} (\Sigma - D) + \tilde{U} (2D - D\Sigma^{-1} D)^{1/2},
    \]
    where \(\tilde{U} \in \mathbb{R}^{n \times m}\) has orthonormal columns orthogonal to the column space of \(X\), and the matrix square root exists due to specific conditions on \(D\) \citep{barber2015controlling}.

    \item \textit{Augmented model fitting}: Augmenting the knockoff matrix to the original design matrix we obtain $(X,\tilde{X})$. For each \(j = 1, 2, \ldots, m\), the coefficients corresponding to \(X_j\) and \(\tilde{X}_j\) are estimated using a regularized method such as Lasso.

    \item \textit{Feature importance statistics}: Define \(L_j\) as a statistic measuring the importance of \(X_j\), commonly the regularization level at which it enters the Lasso path. Fitting the augmented model yields \(m\) pairs \((L_j, \tilde{L}_j)\).

    \item \textit{Knockoff statistics}: For each \(j = 1, 2, \ldots, m\), define
    \[
    V_j = (L_j \vee \tilde{L}_j)\left[2\mathbb{I}\{L_j > \tilde{L}_j\} - 1\right].
    \]

    \item \textit{Thresholding and selection}: Hypothesis \(H_j\) is rejected if \(V_j > T\), where
    \[
    T = \min\left\{t \in \mathcal{V} : \frac{1 + \text{Card}\{j : V_j \leq -t\}}{\text{Card}\{j : V_j \geq t\} \vee 1} \leq \alpha \right\},
    \]
    or \(T = +\infty\) if the set is empty. Here, \(\mathcal{V} = \{|V_j| : j = 1, 2, \ldots, m\} \setminus \{0\}\).
\end{itemize}
This procedure controls the FDR at the target level \(\alpha\). Extensions to the regime \(m < n \leq 2m\) are also discussed in \citet{barber2015controlling}.

\textbf{Knockoff-assisted methods:} We now describe two knockoff-assisted multiple testing procedures proposed by \citet{sarkar2022adjusting}, which modify the BH procedure using a preliminary screening step and apply it at a reduced significance level.

\begin{itemize}
    \item \textbf{Bonferroni–Benjamini–Hochberg procedure}: This procedure computes
    \[
    \tilde{Q}_j = \begin{cases}
    1, & \text{if } P^{(1)}_j > \sqrt{\alpha}, \\
    P^{(2)}_j, & \text{if } P^{(1)}_j \leq \sqrt{\alpha},
    \end{cases}
    \quad \text{for } j = 1, 2, \dots, m.
    \]
    The BH procedure is then applied to the vector \(\tilde{Q} = (\tilde{Q}_1, \dots, \tilde{Q}_m)\) at the reduced level \(\sqrt{\alpha}\).

    \item \textbf{Adaptive Bonferroni–Benjamini–Hochberg procedure}: This extension estimates the proportion of true null hypotheses \(\pi_0 = m_0 / m\), where \(m_0 = \text{Card}\{j \in \{1, 2, \ldots, m\} : \beta_j = 0\}\). Following \citet{storey2004strong}, \(\pi_0\) is estimated by
    \[
    \hat{\pi}_0 = \frac{1 + \sum_{j=1}^m \mathbb{I}(P_j^{(2)} > \lambda)}{m(1 - \lambda)},
    \]
    for some tuning parameter \(\lambda \in (\sqrt{\alpha}, 1)\). The modified $p$-values are
    \[
    Q^*_j = \begin{cases}
    1, & \text{if } P^{(1)}_j > \sqrt{\alpha}, \\
    \hat{\pi}_0 P^{(2)}_j, & \text{if } P^{(1)}_j \leq \sqrt{\alpha},
    \end{cases}
    \quad \text{for } j = 1, 2, \dots, m.
    \]
    The BH procedure is then applied to the vector \(Q^* = (Q^*_1, \dots, Q^*_m)\) at level \(\sqrt{\alpha}\).
\end{itemize}
 
\section{Proposed methods}

The notations used in this section are consistent with Section~1 and Section~2.2.
From the first set of $p$-values $P^{(1)}$, we obtain the corresponding $e$-values $S^{(1)} = (S^{(1)}_1, S^{(1)}_2, \ldots, S^{(1)}_m)$, where $S^{(1)}_j = g(P^{(1)}_j)$ for $j = 1, 2, \ldots, m$, and $g$ is a $p$-to-$e$ calibrator \citep{vovk2021evalues}.
We discuss $p$-to-$e$ calibrators, their properties, and suggest a suitable one for our purpose in Section~4.

The proposed methods follow a common structure. Information from the first set of $p$-values is transformed to $e$-values, which are then used to modulate the significance of the second-stage $p$-values. The three methods differ in how this information is incorporated and whether adaptation to the proportion of true null hypotheses is employed. We present the methods in increasing order of complexity, starting from a basic $ep$-BH procedure and moving toward adaptive and weighted extensions.

We begin with a baseline procedure that directly applies the $ep$-BH idea. The $e$-values constructed from the first set of $p$-values act as prior weights for the second set of $p$-values. Intuitively, hypotheses supported by stronger first-stage evidence receive greater emphasis, while the classical BH procedure is applied for decision making. We identify this procedure by \textit{Method 1}. The detailed steps of the procedure are given below in Algorithm~1.

\vspace{0.4em}
\hrule height 1pt
\vspace{0.25em}
\noindent\textbf{Algorithm 1: Screening based on $ep$-BH procedure}
\vspace{0.35em}
\hrule height 1pt
\vspace{0.3em}
\begin{enumerate}
\item \textbf{Input:} Two sets of $p$-values, $\{P_j^{(1)}\}$ and $\{P_j^{(2)}\}$ for $j = 1, \dots, m$, and FDR level $\alpha \in (0, 1)$.

\item Compute $e$-values using the calibrator $g$:
\[
S_j^{(1)} = g\!\left(P_j^{(1)}\right), \qquad j = 1, \dots, m.
\]

\item Construct the adjusted $p$-values:
\[
\tilde{P}_j = \frac{P_j^{(2)}}{S_j^{(1)}}, \qquad j = 1, \dots, m.
\]

\item Sort the $\tilde{P}_j$'s in increasing order:
\[
\tilde{P}_{(1)} \leq \tilde{P}_{(2)} \leq \cdots \leq \tilde{P}_{(m)}.
\]
Let $H_{(j)}$ be the null hypothesis corresponding to $\tilde{P}_{(j)}$ for each $j = 1,2,\dots,m$.

\item Determine the largest index
\[
\tilde{R}
= \max \left\{ j \in \{1, \dots, m\} : \tilde{P}_{(j)} \leq \frac{j\alpha}{m} \right\},
\]
if such a $j$ exists; otherwise, set $\tilde{R} = 0$.

\item \textbf{Output:} Reject the null hypotheses $H_{(j)}$ for all $j \leq \tilde{R}$.
\end{enumerate}
\vspace{0.4em}
\hrule height 1.2pt
\vspace{0.5em}

Summarily, \textit{Method 1} is the unnormalized $e$-value weighted BH procedure based on the work of \citet{ignatiadis2024evalues}. It uses the ratio of $P_j^{(2)}$ to $S_j^{(1)}$ and applies the classical BH cutoff for decision making. 

\textit{Method~1} does not require estimating the proportion of true null hypotheses and therefore it has exact FDR control under the model considered. However, as with the classical BH procedure, this method can be conservative when the fraction of non-null hypotheses is small.

To address the potential conservativeness of \textit{Method~1}, we next consider an adaptive extension in \textit{Method~2}. The key idea is to estimate the proportion of true null hypotheses using the second set of $p$-values and incorporate this estimate into the BH method. This leads to an adaptive $ep$-BH method that can improve power while retaining asymptotic FDR control. The detailed steps of \textit{Method~2} are given below in Algorithm~2. 

\vspace{0.4em}
\hrule height 1pt
\vspace{0.25em}
\noindent\textbf{Algorithm 2: Screening based on adaptive $ep$-BH procedure}
\vspace{0.35em}
\hrule height 1pt
\vspace{0.3em}

\begin{enumerate}
\item \textbf{Input:} Two sets of $p$-values, $\{P_j^{(1)}\}$ and $\{P_j^{(2)}\}$ for $j = 1, \dots, m$, FDR level $\alpha \in (0, 1)$, and a tuning parameter $\lambda \in (0, 1)$.

\item Estimate the proportion of true null hypotheses:
\[
\hat{\pi}_0
= \frac{1 + \sum_{j=1}^m \mathbb{I}\!\left(P_j^{(2)} > \lambda\right)}{m(1 - \lambda)}.
\]

\item Compute $e$-values using the calibrator $g$:
\[
S_j^{(1)} = g\!\left(P_j^{(1)}\right), \qquad j = 1, \dots, m.
\]

\item Compute adjusted $p$-values:
\[
P_j^*
= \frac{\hat{\pi}_0 \cdot P_j^{(2)}}{\mathbb{I}\!\left(P_j^{(2)} \leq \lambda\right) \cdot S_j^{(1)}},
\qquad j = 1, \dots, m.
\]

\item Sort the $P_j^*$'s in increasing order:
\[
P_{(1)}^* \leq P_{(2)}^* \leq \cdots \leq P_{(m)}^*.
\]
Let $H_{(j)}$ be the null hypothesis corresponding to $P_{(j)}^*$ for each $j = 1,2,\dots,m$.

\item Determine the largest index
\[
R^*
= \max \left\{ j \in \{1, \dots, m\} : P_{(j)}^* \leq \frac{j\alpha}{m} \right\},
\]
if such a $j$ exists; otherwise, set $R^* = 0$.

\item \textbf{Output:} Reject the null hypotheses $H_{(j)}$ for all $j \leq R^*$.
\end{enumerate}

\vspace{0.4em}
\hrule height 1.2pt
\vspace{0.5em}

Due to the additional step involving the estimation of $\pi_0$ in \textit{Method~2}, exact finite-sample FDR control cannot be established theoretically under the considered framework, although simulation results indicate that the FDR remains controlled in finite samples. Nevertheless, under the framework considered here, consistency of the variance estimator ensures that asymptotic FDR control is preserved.

In Method~3, we first standardize the $e$-values obtained from the first set of $p$-values so that the resulting weights sum to $m$. These standardized $e$-values are then used as $p$-value weights. The proportion of true null hypotheses is estimated using a weighted version of Storey’s estimator, which incorporates these weights. Finally, an adaptive weighted BH procedure is applied to the second set of $p$-values using the estimated null proportion and the constructed weights. The detailed steps of \textit{Method~3} are given below in Algorithm~3.

\vspace{0.4em}
\hrule height 1pt
\vspace{0.25em}
\noindent\textbf{Algorithm 3: Screening based on adaptive weighted BH procedure}
\vspace{0.35em}
\hrule height 1pt
\vspace{0.3em}

\begin{enumerate}
\item \textbf{Input:} The $p$-values $\{P_j^{(1)}\}, \{P_j^{(2)}\}$ for $j = 1, \dots, m$, FDR level $\alpha \in (0, 1)$, and a tuning parameter $\lambda \in (0, 1)$.

\item Compute weights based on the calibrator:
\[
W_j^{(1)} = \frac{m \cdot S_j^{(1)}}{\sum_{k=1}^m S_k^{(1)}},
\qquad j = 1, \dots, m.
\]

\item Estimate the proportion of null hypotheses:
\[
\hat{\delta}_0
= \frac{\max_j W_j^{(1)} + \sum_{j=1}^m W_j^{(1)} \cdot \mathbb{I}\!\left(P_j^{(2)} > \lambda\right)}{m(1 - \lambda)}.
\]

\item Compute $e$-values using the calibrator $g$:
\[
S_j^{(1)} = g\!\left(P_j^{(1)}\right), \qquad j = 1, \dots, m.
\]

\item Compute the weighted adjusted $p$-values:
\[
P_j^+
= \frac{\hat{\delta}_0 \cdot P_j^{(2)}}{W_j^{(1)}},
\qquad j = 1, \dots, m.
\]

\item Sort the $P_j^+$ values:
\[
P_{(1)}^+ \leq P_{(2)}^+ \leq \cdots \leq P_{(m)}^+.
\]
Let $H_{[j]}$ be the null hypothesis corresponding to $P_{(j)}^+$ for each $j = 1,2,\dots,m$.

\item Identify the largest index
\[
R^+
= \max \left\{ j \in \{1, \dots, m\} :
P_{(j)}^+ \leq \min\!\left( \hat{\delta}_0 \lambda, \frac{j\alpha}{m} \right) \right\},
\]
if such a $j$ exists; otherwise, set $R^+ = 0$.

\item \textbf{Output:} Reject the null hypotheses $H_{[j]}$ for all $j \leq R^+$.
\end{enumerate}

\vspace{0.4em}
\hrule height 1.2pt
\vspace{0.5em}

Although the construction of \textit{Method~3} differs from that of \textit{Method~2}, the two methods are expected to exhibit comparable performance in most of the practical settings. Differences may arise when the $e$-values obtained from the first set of $p$-values show substantial heterogeneity across hypotheses, in which case the normalization and weighting scheme in \textit{Method~3} can lead to slightly different prioritization of hypotheses compared to the direct scaling used in \textit{Method~2}. \textit{Method~3} also admits only asymptotic FDR control under the framework considered here. Simulation studies indicate that \textit{Method~3} maintains satisfactory FDR control even in finite samples.

The proposed methods are designed to control the false discovery rate under the modeling assumptions described in Section~1. Theorem~\ref{thm:combined_theorem} summarizes the FDR control properties of the proposed methods.

\begin{theorem}
\label{thm:combined_theorem}
Suppose that the model in (\ref{eq:regmodel}) holds. Then:
\begin{itemize}
    \item \textit{Method 1} controls the FDR at level $\pi_0 \alpha$.
    \item \textit{Method 2} asymptotically controls the FDR at level $\alpha$, that is, the FDR is bounded by $\alpha + o(1)$ as $\nu \to \infty$.
    \item \textit{Method 3} also asymptotically controls the FDR at level $\alpha$, with $o(1) \to 0$ as $\nu \to \infty$.
\end{itemize}
\end{theorem}

\begin{proof}[Proof of Theorem~\ref{thm:combined_theorem}]
Provided in \emph{Appendix A}.
\end{proof}

\textbf{Remark 1.} When $\sigma^2$ is assumed to be known, the second set of $p$-values $P^{(2)}$ are independent among themselves and it is independent with the first set of $p$-values $P^{(1)}$. Consequently, $P^{(2)}$ is also independent of $S^{(1)}$ and $W^{(1)}$. Thus, the proof of Theorem~\ref{thm:combined_theorem} directly follow from Theorem 4.2 and Theorem 4.10 of \citet{ignatiadis2024evalues}, and from Theorem 3.1 of \citet{biswas2023new}.

\textbf{Remark 2.} The choice among the three proposed methods depends on the practitioner’s objective, FDR threshold to use and the expected signal structure. \textit{Method~1} guarantees exact FDR control in finite samples, and is therefore preferable when strict error control is a priority. When a moderate to large number of non-null hypotheses is expected, \textit{Method~2} or \textit{Method~3} may be used to improve power. \textit{Methods~2} and\textit{~3} exhibit largely comparable performance. However, in settings with a very large number of hypotheses, higher FDR threshold and sparse non-null signals, \textit{Method~3} can screen slightly more discoveries than \textit{Method~2}. This behavior may be attributed to adjustment of heterogeneity through standardization of the $e$-values and the use of weighted null proportion estimation in \textit{Method~3}.

\section{Bounded calibration}

\textbf{Calibration and admissibility:} As mentioned in Section~1, \(p\)-values can be transformed into \(e\)-values using \(p\)-to-\(e\) calibrators. A \(p\)-to-\(e\) calibrator is a decreasing function \(g : [0,1] \to [0,\infty]\) satisfying \(\int_0^1 g(t)\,dt \leq 1\). A calibrator \(g\) is called admissible if there does not exist another function \(h\) such that \(h(t) \geq g(t)\) for all \(t \in [0,1]\), with strict inequality for at least one \(t\). \citet{vovk2021evalues} show that \(g\) is admissible if \(g(0) = \infty\) and \(\int_0^1 g(t)\,dt = 1\). One such admissible calibrator is \(g(t) = t^{-1/2} - 1\), so that \(S = P^{-1/2} - 1\) is an \(e\)-value if \(P\) is a \(p\)-value. Another admissible family is \(g_\kappa(t) = \kappa t^{\kappa - 1}\) for \(\kappa \in (0,1)\), as well as its averaged version \(\int_0^1 \kappa t^{\kappa - 1} d\kappa\). Admissible calibrators are theoretically preferred in classical testing problems because they often correspond to likelihood-ratio-type statistics and are optimal in the sense that no other calibrator dominates them point-wise while preserving validity.

\textbf{Use of bounded calibrators:} Our context involves a meta-analysis-type multiple testing problem, where the \(e\)-values derived from the first set of \(p\)-values are not directly used for decision-making but rather act as weights to modulate the significance of the second set of \(p\)-values. In this setup, using bounded calibrators is not only reasonable but often desirable for the following reasons:
\begin{itemize}
\item \textit{Boundedness ensures stability:} In simulation experiments, we observed that the bounded calibrators avoid extreme values, which can otherwise introduce variance inflation in the combined \(p\)-values.
\item \textit{Variance-based justification:} Since \(\hat{\beta}^{(1)}\) tends to have higher variance than \(\hat{\beta}^{(2)}\), placing primary emphasis on \(P_j^{(1)}\) via large \(e\)-values may mislead the procedure. Instead, we use \(S_j^{(1)}\) as a controlled and supportive weight, letting \(P_j^{(2)}\) — based on more stable estimators — drive the multiple testing.
\item \textit{Procedure design:} The use of bounded calibrators enables the formulation of more refined strategies, such as \textit{Method 3}, which blends information across stages in a stable yet adaptive manner.
\end{itemize}
Therefore, while admissible calibrators are of greater importance for testing single hypotheses, in our current multiple testing setup — particularly when integrating two sources of evidence — the use of bounded admissible-like calibrators serves both theoretical consistency and practical robustness. 

\textbf{Bounded admissibility:} A bounded $p$-to-$e$ calibrator is a decreasing function $g:[0,1]\to[0,C]$ satisfying $\int_0^1 h(t)dt\leq 1$. A $p$-to-$e$ calibrator $g$ is defined to be bounded-admissible if it is bounded and there does not exist another $h:[0,1]\to[0,C]$ such that $h(t)\geq g(t)$ for all $t\in [0,1]$, with strict inequality for at least one $t$. Following \citet{vovk2021evalues}, it is easy to show that $g$ is bounded-admissible if $g(0)=C$ and $\int_0^1 g(t) dt=1$.

\textbf{A special case:} Here, we identify a specific calibrator that reveals the methods proposed by \citet{sarkar2022adjusting} are special cases of our general methods introduced in Section~4 under this particular choice.
Consider the bounded-admissible all-or-nothing calibrator  
\[
g_1(t) = \frac{1}{\sqrt{\alpha}} \cdot \mathbb{I}(t \leq \sqrt{\alpha}),
\]
which maps the first-stage \(p\)-values \(\{P_j^{(1)}\}_{j=1}^m\) to \(e\)-values \(S_j^{(1)} = g_1(P_j^{(1)})\). This function is bounded with range \([0, 1/\sqrt{\alpha}]\) and satisfies the admissibility condition of our framework.

Using this specific calibrator within \textit{Method 1}, the resulting transformed \(p\)-values coincide with the Bon-BH method:
\(
\tilde{P}_j = \tilde{Q}_j, \quad \text{for } j = 1, \dots, m,
\)
where \(\tilde{Q}_j\) is defined as in Section~2.2. Thus, Bon-BH arises as a special case of our general method.

Similarly, applying \textit{Method 2} with the same \(g_1\), and estimating \(\pi_0\) as in Storey’s method, the final adjusted \(p\)-values satisfy
\(
P_j^* = Q_j^*, \quad \text{for } j = 1, \dots, m,
\)
recovering the adaptive Bon-BH procedure described earlier.

More generally, by choosing a family of all-or-nothing calibrators of the form  
\[
g_r(t) = \frac{1}{\alpha^r} \cdot \mathbb{I}(t \leq \alpha^r), \quad \text{for any } r \in (0,1),
\]
one can recover the broader class of Bon-BH-type methods discussed in the supplementary material of \citet{sarkar2022adjusting}.

\textbf{Prescribed choice of calibrator:} An example of a continuous and bounded-admissible $p$-to-$e$ calibrator is given by
\[
g_2(t) = C(1 - t^a),
\]
where \(C > 1\) and \(a > 0\). To satisfy the admissibility (integral) condition, we set \(a = 1/(C - 1)\). We recommend using \(g_2\) in \textit{Method 1}, \textit{Method 2}, and \textit{Method 3}, particularly with
\(
C = 1/\alpha.
\)
This choice ensures that \(g_2\) satisfies the admissibility criterion while remaining continuous, enabling a smoother and more informative transfer of weights to the second set of $p$-values. Intuitively, when the first-stage $p$-value is close to one, the corresponding $e$-value is close to zero, and hence the hypothesis receives little support for rejection. On the other hand, when the first-stage $p$-value is close to zero, the resulting $e$-value approaches \(C\), so that the second-stage $p$-value is effectively compared against a threshold of order \(C\alpha\). Choosing \(C = 1/\alpha\) therefore allows a hypothesis with strong first-stage evidence to be rejected whenever the second-stage $p$-value is at most one, which is consistent with intuition from single-hypothesis testing. The constant \(C\) thus scales the $e$-values to ensure proper calibration under the target significance level \(\alpha\). This flexible structure has shown visible improvements over the Bon-BH procedure in both simulations and real data applications, making it a practically effective and theoretically sound alternative.

\section{Simulation study}

In our simulation setting, We consider a linear regression model with $n$ observations and $m$ covariates. Among the $m$ covariates, $k$ are relevant (non-null), and the remaining $m-k$ are null. The $n \times m$ design matrix $X$ is generated from a multivariate normal distribution with mean vector $0$ and an autoregressive correlation structure of order $1$ (AR(1)). Specifically, the covariance matrix $\Omega \in \mathbb{R}^{m \times m}$ has entries $\Omega_{ij} = \rho^{|i - j|}$ for some fixed $\rho \in (0, 1)$. Each row of $X$ is sampled independently from $N_m(0, \Omega)$. The columns of $X$ are then standardized to unit norm to form the matrix $\mathbf{Z}$, i.e.,
\[
Z_{ij} = \frac{X_{ij}}{\sqrt{(X^\top X)_{jj}}}.
\] 
The true coefficient vector $\beta \in \mathbb{R}^m$ is generated such that exactly $k$ coordinates are non-zero, each having magnitude $\gamma > 0$, and the remaining are set to zero. The locations of the non-zero entries are selected uniformly at random. Given $Z$ and $\beta$, the response vector $Y \in \mathbb{R}^n$ is generated according to the linear model
\[
Y_i = Z_i^\top \beta + \varepsilon_i, \quad \varepsilon_i \sim N (0, \sigma^2),
\]
independently for $i = 1, \dots, n$. We compare the following multiple testing-based variable selection procedures:
\begin{itemize}
    \item $M_0$: The original knockoff-filter method. The knockoff design matrix is constructed by using the semi-definite programming approach discussed in \citet{barber2015controlling}. We use the same knockoff design matrix for the following methods. For the knockoff method, we use the lasso penalty parameter for estimating the threshold.
    \item $M_1$: The Bon-BH method, as proposed by \citet{sarkar2022adjusting}.
    \item $M_2$: The adaptive Bon-BH method, as proposed by \citet{sarkar2022adjusting}.
    \item $M_3$: \textit{Method 1} of Section~3 with the calibrator $g_2$ mentioned in Section~4.
    \item $M_4$: \textit{Method 2} of Section~3 with the calibrator $g_2$ mentioned in Section~4.
    \item $M_5$: \textit{Method 3} of Section~3 with the calibrator $g_2$ mentioned in Section~4.
\end{itemize}

\begin{figure}[p]  
    \centering
    \includegraphics[width=\textwidth]{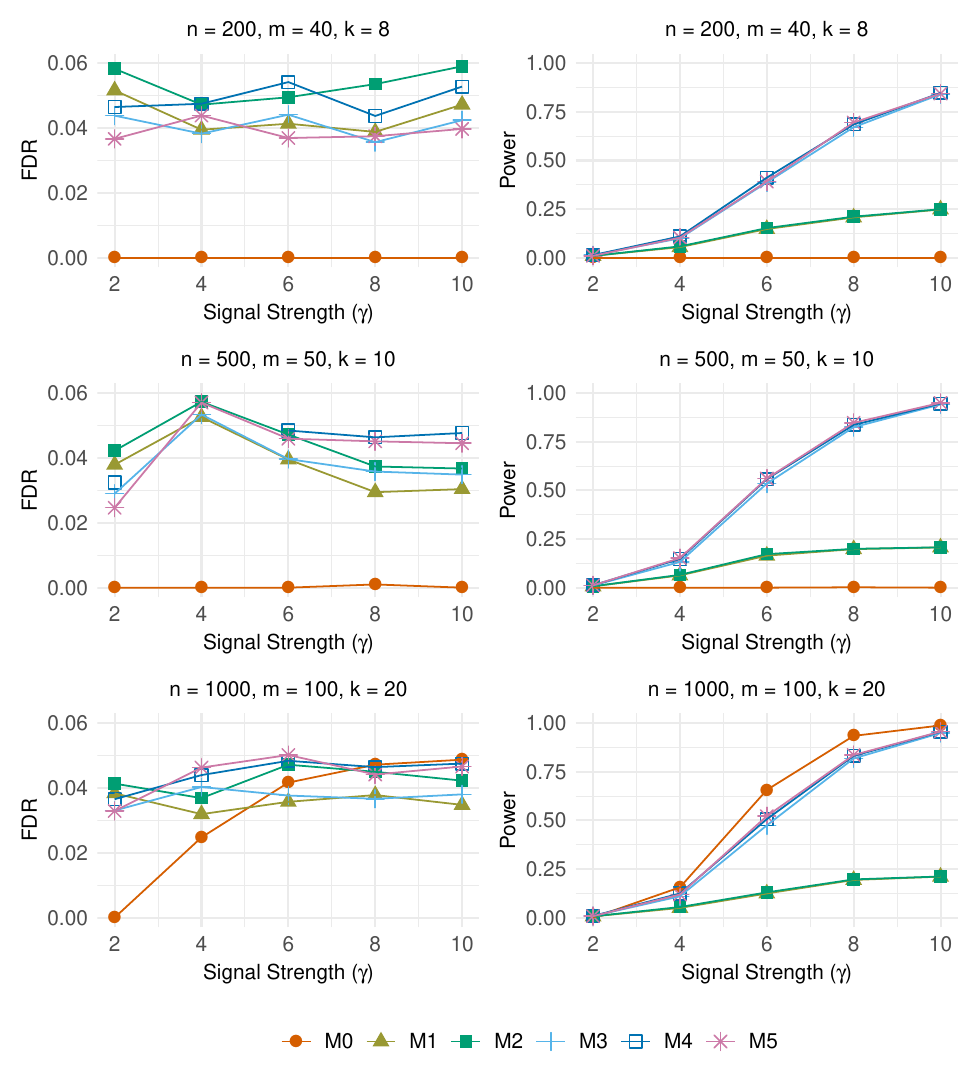}
    \caption{
        Comparison of false discovery rate and power across different settings $(n, m, k)$ with $\rho=0.5$ and $\alpha = 0.05$. Each row corresponds to a different simulation setting, and signal strength values ($\gamma = 2, 4, 6, 8, 10$) are shown on the x-axis.
    }
    \label{fig:simul_0.05}
\end{figure}

\begin{figure}[p]  
    \centering
    \includegraphics[width=\textwidth]{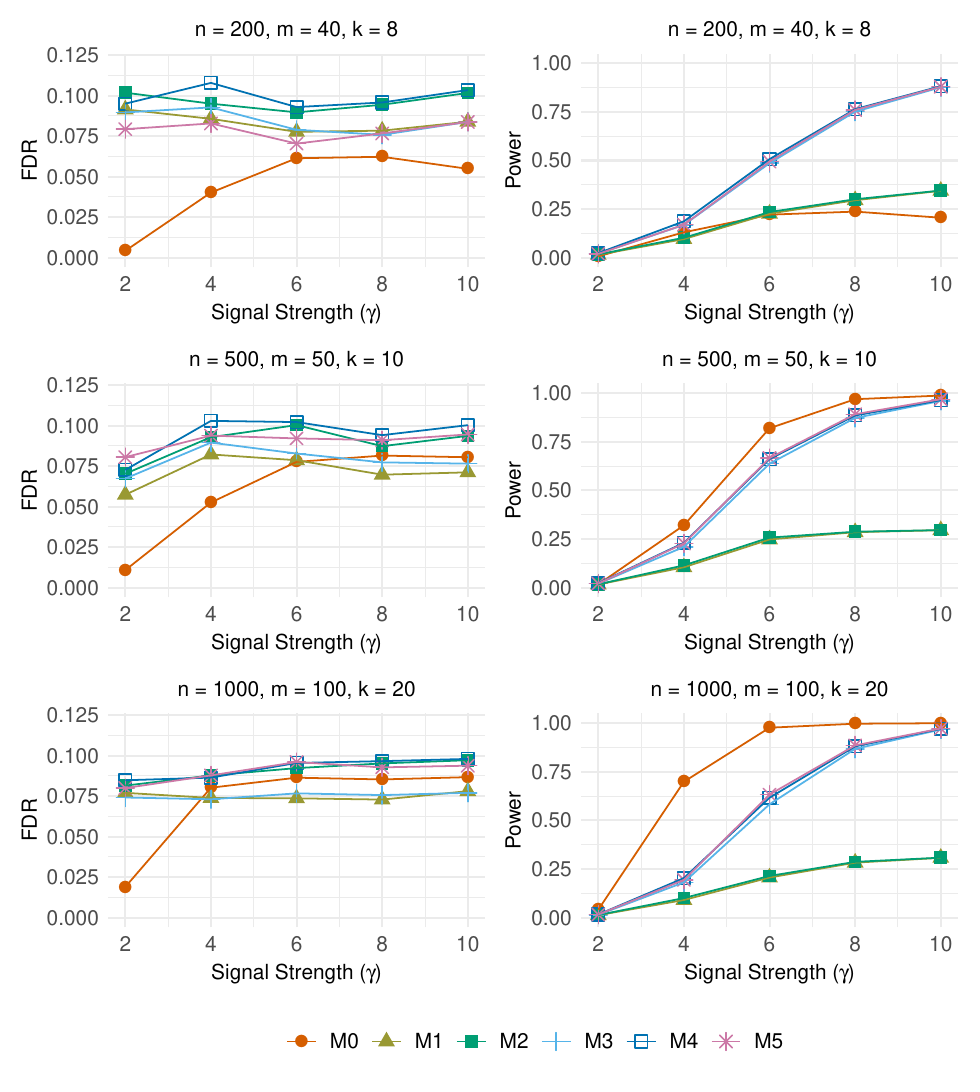}
    \caption{
        Comparison of false discovery rate and power across different settings $(n, m, k)$ with $\rho=0.5$ and $\alpha = 0.1$. Each row corresponds to a different simulation setting, and signal strength values ($\gamma = 2, 4, 6, 8, 10$) are shown on the x-axis.
    }
    \label{fig:simul_0.1}
\end{figure}

We implement the above methods with the level of FDR prefixed at $\alpha$. To assess performance, we compute the following quantities for each method:
\begin{itemize}
    \item Empirical false discovery rate: The average proportion of false positives among all discoveries, across all simulations.
    \item Empirical power: The average proportion of true positives among all truly non-null features, across all simulations.
\end{itemize}
Each method is evaluated over $N = 500$ repetitions of the simulation, and the mean empirical FDR and power are reported.

We set $(n,m,k) \in \{(200,40,8),\ (500,50,10),\ (1000,100,20)\}$ with $\rho = 0.5$, and $\lambda = 0.5$ for computing $\hat{\pi}_0$, and considered $\gamma \in \{2,4,6,8,10\}$ and $\alpha \in \{0.05, 0.1\}$. Simulation results for these settings are presented in Figures~\ref{fig:simul_0.05} and \ref{fig:simul_0.1}. We also experimented with several other configurations to assess the robustness of our findings. In particular, we additionally considered $(n,m,k) \in \{(200,40,20),\ (500,50,25),\ (1000,100,50)\}$ with $\rho \in \{0.1, 0.9\}$, corresponding to different proportions of true null hypotheses ($0.8$, and $0.5$) and varying levels of multicollinearity. The simulation results for these additional settings are presented in the \textit{Appendix B}. Across all scenarios considered, the qualitative conclusions remain consistent.

It is evident from Figures~\ref{fig:simul_0.05} and \ref{fig:simul_0.1} that all the methods maintain satisfactory control of the FDR at their respective nominal levels, $\alpha = 0.05$ and $\alpha = 0.1$, across varying simulation configurations. The proposed method $M_3$ consistently ensures strict FDR control, as supported by theoretical guarantees, while $M_4$ and $M_5$ also maintain strong empirical FDR control across both levels of $\alpha$. When the signal strength is low (e.g., $a = 2$), the knockoff method $M_0$ remains highly conservative in both figures, yielding negligible FDR and very low power, a trend that becomes less pronounced as the signal strengthens. Additional simulations show that in non-sparse settings with a moderate proportion of true nulls (here \(0.5\)) and low multicollinearity (\(\rho = 0.1\)), \(M_0\) is not conservative and, in fact, achieves the highest power. The methods $M_1$ and $M_2$ show slightly higher FDR than $M_0$, with improved power, yet they still fall short compared to the proposed methods. Notably, when $\alpha = 0.1$, all methods benefit from increased flexibility, leading to visibly higher power, but the relative performance patterns remain consistent. The power of $M_3$, $M_4$, and $M_5$ improves substantially, especially for moderate to high signal strengths, while still preserving acceptable FDR control. $M_4$ and $M_5$ continue to exhibit nearly indistinguishable power curves, with $M_5$ occasionally exhibiting a marginal edge. $M_3$ remains slightly conservative in power, particularly in weaker signal scenarios, yet demonstrates robust overall performance. Interestingly, in the strong signal regime (e.g., $a \geq 6$), particularly for $n = 1000$, $m = 100$, and $k = 20$, the knockoff method $M_0$ achieves surprisingly high power under both $\alpha$ levels, underscoring its potential in large-sample, high-signal contexts. Even so, the proposed methods $M_3$, $M_4$, and $M_5$ remain competitive and preferable due to their consistent balance between FDR control and statistical power. These trends are stable across both figures, indicating that the relative strengths of the proposed methods over the alternatives are not merely specific to a single FDR level. It is also worth noting that many practical applications involve low signal strengths and a small number of truly important hypotheses - settings where $M_3$, $M_4$, and $M_5$ clearly outperform the alternatives. Based on the overall evaluation, $M_3$, $M_4$, and $M_5$ emerge as the most reliable methods for robust FDR control with strong power across diverse scenarios. It is also observed that the overall performance of all methods deteriorates as multicollinearity increases, although their relative performance remains similar. We recommend using the proposed methods instead of the original knockoff procedure when the problem is sparse, with few and weak signals, which is common in practice; otherwise, the original knockoff method remains a better choice.

\section{Data analysis}

We analyze the HIV-1 drug resistance dataset to evaluate the performance of our proposed variable selection methods in a real-world genomics setting. The dataset, originally compiled and analyzed by \citet{rhee2005hiv, rhee2006genotypic}, includes measurements of drug resistance and genotype information for HIV-1 samples. Our analysis focuses on protease inhibitor (PI) drugs, which are a class of antiretroviral treatments designed to inhibit the HIV-1 protease enzyme and are widely used in clinical management of HIV.

\begin{figure}[h!]  
    \centering
    \includegraphics[width=\textwidth]{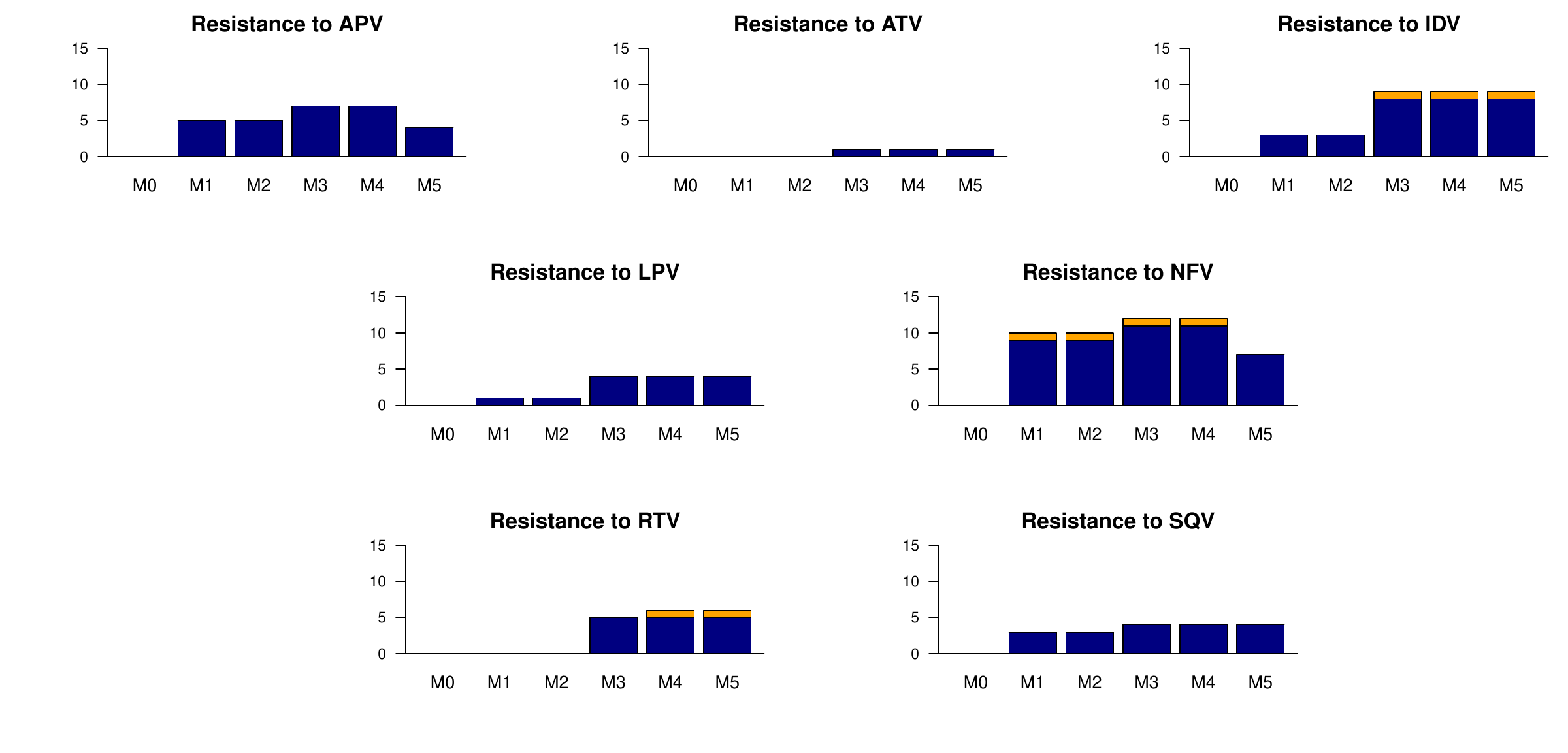}
    \caption{Variable selection results at false discovery rate level $\alpha = 0.05$ for each drug. Blue bars indicate positions with prior biological support; orange bars represent novel discoveries.}
    \label{fig:pi_0.05}
\end{figure}

\begin{figure}[h!]  
    \centering
    \includegraphics[width=\textwidth]{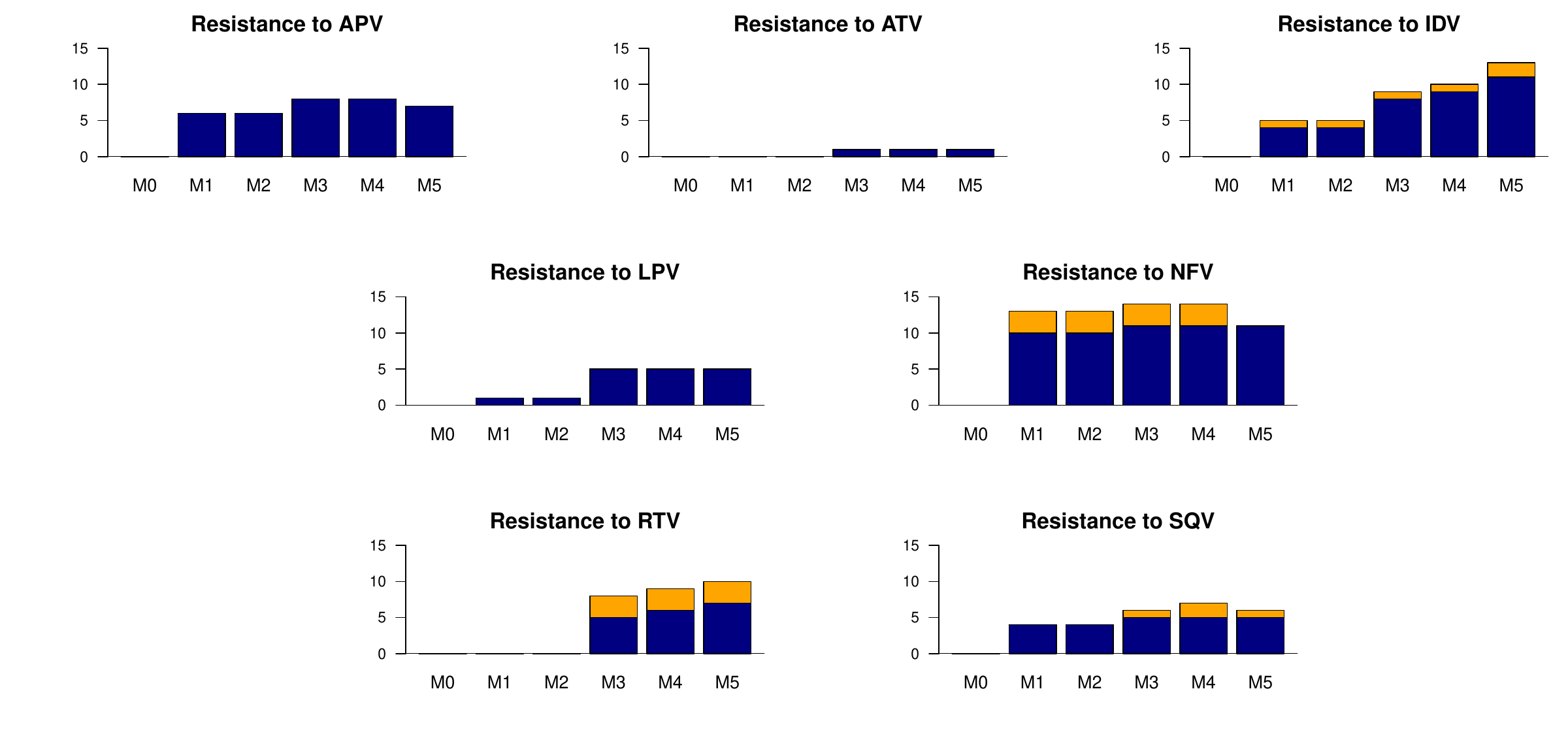}
    \caption{Variable selection results at false discovery rate level $\alpha = 0.1$ for each drug. Blue bars indicate positions with prior biological support; orange bars represent novel discoveries.}
    \label{fig:pi_0.1}
\end{figure}

The PI dataset contains resistance profiles for 7 different drugs across 846 samples, with genetic data available at 99 protease or reverse transcriptase (RT) positions. Among these, 209 distinct mutations were observed at least 3 times across the sample population. For each PI drug, we treat the response variable as the log-fold increase in lab-tested drug resistance for a given sample, and the covariates as binary indicators representing the presence or absence of specific mutations. Following \citet{barber2015controlling}, we preprocess the data by selecting only mutations that appear in at least 3 samples for each drug and removing duplicate columns in the design matrix to ensure identifiability.

Each drug is analyzed separately under an additive linear model with no interaction terms. We apply the knockoff filter ($M_0$) and the BH-type procedures (denoted as methods $M_1$ through $M_5$), each controlling the FDR at level $\alpha = 0.05$ and $\alpha = 0.1$. Here, $M_0$ refers to the standard knockoff filter, $M_1$ and $M_2$ are the methods proposed by \citet{sarkar2022adjusting}, while $M_3$, $M_4$, and $M_5$ are our proposed bounded $e$-value weighted procedures.

To validate the selections, we compare the identified mutation positions to the treatment-selected mutation (TSM) panel established by \citet{rhee2005hiv}. The TSM panel consists of mutation positions found to occur significantly more often in patients treated with PI-class drugs compared to untreated individuals, and serves as a proxy for ground truth in the absence of true labels. 

Figures~\ref{fig:pi_0.05} and~\ref{fig:pi_0.1} summarize the selection results at $\alpha = 0.05$ and $\alpha = 0.1$, respectively. Each panel corresponds to a different PI drug, and the bars indicate the number of mutation positions selected by each method. Blue bars represent positions that are supported by prior biological evidence (i.e., present in the TSM panel), while orange bars correspond to novel discoveries suggested by the method. Our results demonstrate that the proposed methods $M_3$, $M_4$, and $M_5$ outperform $M_0$, $M_1$ and $M_2$ across all drugs at $\alpha = 0.05, 0.1$. Note that the knockoff method does not detect any significant positions in this case, likely due to the small value of $\alpha$. In scenarios with a low prefixed FDR level, the knockoff-adjusted methods-particularly the proposed ones-demonstrate a clear advantage.

\section{Remarks}

Our primary contribution lies in proposing a new class of knockoff-assisted $p$-value-based multiple testing procedures that incorporate $e$-values as weights. This is a novel mechanism for improving power while maintaining control of the FDR. The proposed methods are also applicable in the regime $m < n \leq 2m$, where direct construction of knockoffs requires careful adaptation. Following the construction strategy of \citet{barber2015controlling} and implementation ideas demonstrated by \citet{sarkar2022adjusting}, we note that valid knockoff variables can still be created using row augmentation techniques, which ensures the applicability of our procedures.

Several works have studied optimal weighting strategies for weighted and adaptive weighted BH procedures, including \citet{roeder2009genome}, \citet{Habiger2017adaptive}, and more recently \citet{zheng2024fdr}. While these approaches provide principled ways to construct data-dependent weights under additional modeling assumptions, they do not directly yield a weighting or calibration function that aligns with our current framework. In particular, their optimal weights typically require auxiliary information or prior structure. In our problem setup, $\hat{\beta}^{(1)}$ can serve as the auxilliary information. Exploring such connections is an interesting direction for future research. 

A promising direction for future work is the extension of our methods to generalized linear models (GLMs), where model complexity and variance heterogeneity introduce new challenges. The recent work of \citet{rilling2025new} provides a framework that accommodates such complexities by constructing paired $p$-values, making it a natural foundation for extending our $e$-value-based methodology beyond linear models. While we have suggested a bounded calibrator in this work, it would be of interest to explore alternative calibrators within this framework and investigate the possibility of identifying an optimal choice.  

Finally, while our focus has been on controlling the standard FDR, our general framework naturally invites further exploration into alternative Type-I error metrics, such as the $k$-false discovery rate. As highlighted in \citet{sarkar2022adjusting}, such extensions are important and feasible, suggesting that the methods proposed here possess both robustness and adaptability to broader multiple testing settings.

\subsection*{Declaration} 
The authors declare no conflict of interest.

The dataset used in this study is openly available. All computations and visualizations were performed using the \textsf{R} programming language (\url{https://www.R-project.org/}). The codes can be made available in accordance with journal policies upon request.

\bibliography{Reference}

\newpage

\section*{Appendix A}

First, we state and prove a key lemma that is instrumental in establishing the asymptotic control of the false discovery rate (FDR) for \textit{Method~2} and \textit{Method~3} described in the main manuscript.

\begin{lemma}
\label{lem:asymptotic_pval}
Let $Y=X\beta+\varepsilon$, where $Y\in \mathbb{R}^n$ is the response vector, $X \in \mathbb{R}^{n \times m}$ is a fixed full-rank design matrix, $\beta \in \mathbb{R}^m$ is the vector of unknown coefficients, and $\varepsilon \sim N_n(0, \sigma^2 I_n)$ with unknown $\sigma^2 > 0$. Let $\hat{\beta} = (X^\top X)^{-1}X^\top Y$ be the least squares estimator, and define the residual variance estimator as
\[
\hat{\sigma}^2 = \frac{1}{\nu} \|Y - X\hat{\beta}\|^2, \quad \text{where } \nu = n - m.
\]
Consider testing the null hypothesis $H_0: \beta_j = 0$ against the alternative $H_1: \beta_j \neq 0$, $j\in \{1,2,...,m\}$, using the test statistic
\[
T_\nu = \frac{\hat{\beta}_j}{\hat{\sigma} \sqrt{(X^\top X)^{-1}_{jj}}}.
\]
Let the corresponding two-sided \textit{p}-value be
\[
p_\nu = 2\left\{1 - F_\nu\left(|T_\nu|\right)\right\},
\]
where $F_\nu$ is the CDF of the $t$-distribution with $\nu$ degrees of freedom. Then, under $H_0$, as $\nu \to \infty$,
\[
\Pr(p_\nu \leq \alpha \mid \hat{\sigma}) = \alpha + o(1),
\]
for any $\alpha \in [0, 1]$. That is, the \textit{p}-value remains asymptotically valid even when conditioned on $\hat{\sigma}$.
\end{lemma}

\begin{proof}[Proof of Lemma~\ref{lem:asymptotic_pval}]

Under $H_0: \beta_j = 0$, we have:
\[
\hat{\beta}_j \sim N\left(0, \sigma^2 (X^\top X)^{-1}_{jj}\right).
\]
Also,
\[
\frac{\nu \hat{\sigma}^2}{\sigma^2} \sim \chi^2_\nu,
\]
and the numerator and denominator of $T_\nu$ are independent. Hence, under $H_0$
\[
T_\nu \sim t_\nu \quad \Rightarrow \quad \mathbb{P}(p_\nu \leq \alpha) = \alpha.
\]
As $\nu \to \infty$, 
\(
\hat{\sigma}^2 \xrightarrow{P} \sigma^2 \quad \Rightarrow \hat{\sigma} = \sigma + o_P(1)
\).
Thus, under $H_0$, \(
T_\nu = Z \cdot \frac{\sigma}{\hat{\sigma}} = Z + o_P(1).
\), where
\[
Z = \frac{\hat{\beta}_j}{\sigma \sqrt{(X^\top X)^{-1}_{jj}}} \sim N(0,1).
\]
Since $F_\nu(t) = \Phi(t) + r_\nu(t)$ with $r_\nu(t) = o(1)$ uniformly in $t$,
\[
p_\nu = 2(1 - F_\nu(|T_\nu|)) = 2(1 - \Phi(|T_\nu|)) + o(1).
\]
As $|T_\nu| = |Z| + o_P(1)$,
\[
\Phi(|T_\nu|) = \Phi(|Z|) + o_P(1) \quad \Rightarrow \quad p_\nu = 2(1 - \Phi(|Z|)) + o_P(1).
\]
Since $\hat{\sigma} \xrightarrow{P} \sigma$ and $Z \sim \mathcal{N}(0,1)$,
\[
\Pr(p_\nu \leq \alpha \mid \hat{\sigma}) = \Pr(2(1 - \Phi(|Z|)) + o_P(1) \leq \alpha \mid \hat{\sigma}) = \alpha + o(1).
\]
\end{proof}

Now, we prove the main result stated in the manuscript concerning the FDR control of the proposed methods. Specifically, we establish exact FDR control for \textit{Method~1}, and asymptotic FDR control for \textit{Method~2} and \textit{Method~3}.

\begin{proof}[Proof of Theorem~\ref{thm:combined_theorem}] We present the proofs of the three parts separately.

FDR control of \textit{Method~1}: 
Let $\alpha_j=j\alpha/m$ for $j=1,2,...,m$. In \textit{Method~1}, $\tilde{R}$ denotes the number of rejected null hypotheses. For each $j\in \{1,2,...,m\}$, we consider $\tilde{P}_{(i)\setminus\{j\}}$, $i=1,2,...,m-1$ to be the ordered components of $(\tilde{P}_1, \tilde{P}_2, ..., \tilde{P}_m)\setminus\{\tilde{P}_j\}$ and $\tilde{P}_{(0)\setminus\{j\}}=0$. We define $\tilde{R}_{(-j)}$ as:
\[
\{\tilde{R}_{(-j)} = r - 1\} \equiv \{\tilde{P}_{(r-1)\setminus \{j\}} \leq \alpha_r, \tilde{P}_{(r)\setminus \{j\}} \geq \alpha_{r+1}, \ldots, \tilde{P}_{(m-1)\setminus \{j\}} \geq \alpha_m\},
\]
for $r = 1,2, \ldots, m$. Then the FDR of \textit{Method~1} is given by

\begin{align*}
\mathrm{FDR} &= \sum_{j \in \mathcal{N}} E \left\{ \frac{\mathbb{I}\left(\tilde{P}_j \leq \alpha_{\tilde{R}}\right)}{\tilde{R} \vee 1} \right\} \\
&= \sum_{j \in \mathcal{N}} E \left\{ \frac{\mathbb{I}\left(\tilde{P}_j \leq \alpha_{\tilde{R}_{(-j)}+1}\right)}{\tilde{R}_{(-j)} + 1} \right\} \\
&= \sum_{j \in \mathcal{N}} \sum_{r=1}^{m} \frac{1}{r} \Pr\left(\tilde{P}_j \leq \alpha_r,\, \tilde{R}_{(-j)} = r - 1\right) \\
&= \sum_{j \in \mathcal{N}} \sum_{r=1}^{m} \frac{1}{r} E \left\{ \Pr(\tilde{R}_{(-j)} = r - 1 \mid \tilde{P}_j) 
\cdot \mathbb{I}(\tilde{P}_j \leq \alpha_r) \right\} \\
&= \sum_{j \in \mathcal{N}} E_{S^{(1)}} E_{\hat{\sigma}} \left[ \sum_{r=1}^{m} \frac{1}{r} E_{P_j^{(2)}} \Big\{ 
\Pr\left(\tilde{R}_{(-j)} = r - 1 \mid P^{(2)}_j, S^{(1)}, \hat{\sigma}\right) \right. \\
&\qquad\qquad\qquad\qquad\qquad\qquad\quad \left. \cdot\, \mathbb{I}(P^{(2)}_j \leq S^{(1)}_j\alpha_r) \,\Big|\, S^{(1)}, \hat{\sigma} \Big\} \right].
\end{align*}

Now, we rewrite the summation over $r$ as:
\begin{align*}
& \sum_{r=1}^{m} \frac{1}{r} \, E_{P_j^{(2)}} \Bigg\{ 
\left[ 
\Pr\left(\tilde{R}_{(-j)} \geq r - 1 \mid P^{(2)}_j, S^{(1)}, \hat{\sigma} \right)
- 
\Pr\left(\tilde{R}_{(-j)} \geq r \mid P^{(2)}_j, S^{(1)}, \hat{\sigma} \right)
\right] \\
& \qquad \cdot \mathbb{I}\left(P^{(2)}_j \leq S^{(1)}_j\alpha_r,\, P^{(x)}_j \leq \lambda\right)
\, \Bigg| \, S^{(1)}, \hat{\sigma} \Bigg\}\\
& = \sum_{r=1}^{m} \Pr(\tilde{R}_{(-j)} \geq r - 1 \mid P^{(2)}_j, S^{(1)}, \hat{\sigma})\\
& \qquad \cdot E_{P_j^{(2)}}\left[ \frac{\mathbb{I}(P^{(2)}_j \leq S^{(1)}_j\alpha_r)}{r} - \frac{\mathbb{I}(P^{(2)}_j \leq S^{(1)}_j\alpha_{r-1})}{r - 1} \Bigg| S^{(1)},\hat{\sigma}\right].
\end{align*}

Conditional on $\hat{\sigma}$, the distribution of $P_j^{(2)}$ is independent of $S^{(1)}$. Let $F$ be the common cumulative distribution function (cdf) of $P_j^{(2)}$ for each $j\in \mathcal{N}$. Thus, 

\begin{align*}
\mathrm{FDR} = \sum_{j \in \mathcal{N}} E_{S^{(1)}} \Bigg[ E_{\hat{\sigma}} \Bigg\{ 
& \sum_{r=1}^{m} \Pr\left(\tilde{R}_{(-j)} \geq r - 1 \mid P^{(2)}_j, S^{(1)}, \hat{\sigma}\right) \\
& \cdot \left( \frac{F(S^{(1)}_j\alpha_r\mid \hat{\sigma})}{r} - \frac{F(S^{(1)}_j\alpha_{r-1}\mid \hat{\sigma})}{r - 1} \right)
\Bigg\} \Bigg]
\end{align*}

For $r=1$, the term within the brackets after the operator $E_{\hat{\sigma}}$ equals
\[
E_{\hat{\sigma}}\left\{F(S_j^{(1)}\alpha_1\mid \hat{\sigma}))\right\},
\]
and for $r=2,3,...,m$, it can be written as 
\[
S^{(1)}_j\alpha_r E_{\hat{\sigma}^*}\left\{\psi(S^{(1)},\hat{\sigma}^*)\left(\frac{1}{r}-\frac{1}{r-1}\frac{F(S^{(1)}_j\alpha_{r-1}\mid \hat{\sigma}^*)}{F(S^{(1)}_j\alpha_r\mid \hat{\sigma}^*)}\right)\right\}.
\]
Here, density of $\hat{\sigma}^*\mid S^{(1)}$ at $z$ is
\[
h^*(z)=\frac{h(z)F(S^{(1)}_j\alpha_r\mid \hat{\sigma}=z)}{S^{(1)}_j\alpha_r},
\]
where $h$ is the density of $\hat{\sigma}$. 

In Section~S.2.4 of the supplementary material of \citet{sarkar2022adjusting}, it is justified that the density of $P_j^{(2)}\mid \hat{\sigma}=z$ at $u$ is totally positive of order 2 (TP$_2$) in $(u,z)$. For the definition of this special type of dependence and some associated discussions under the setup considered in this work, we refer to Section~S.2.1 of the supplementary material of \citet{sarkar2022adjusting}. Thus, 
\[
\left(\frac{1}{r}-\frac{1}{r-1}\frac{F(S^{(1)}_j\alpha_{r-1}\mid \hat{\sigma}^*=z)}{F^*(S^{(1)}_j\alpha_r\mid \hat{\sigma}^*=z)}\right)
\]
is increasing in $z$ and $\psi(S^{(1)},\hat{\sigma}^*=z)$ is decreasing in $z$. This is a consequence of the TP$_2$ dependence as illustrated in the proof of Theorem~S1~(ii), Section~S.2.2 of \citet{sarkar2022adjusting}. 

Note that, $F$ is the cdf of the uniform distribution over $[0,1]$ and hence
\[
E_{\hat{\sigma}^*}\left(\frac{1}{r}-\frac{1}{r-1}\frac{F(S^{(1)}_j\alpha_{r-1}\mid \hat{\sigma}^*)}{F^*(S^{(1)}_j\alpha_r\mid \hat{\sigma}^*)}\right)=\frac{1}{S^{(1)}_j\alpha_r}E_{\hat{\sigma}}\left( \frac{F(S^{(1)}_j\alpha_r\mid \hat{\sigma})}{r} - \frac{F(S^{(1)}_j\alpha_{r-1}\mid \hat{\sigma})}{r - 1} \right)=0.
\]

Using the fact that the covariance of two oppositely monotonic functions is negative, we have
\[
E_{\hat{\sigma}^*}\left\{\psi(S^{(1)},\hat{\sigma}^*)\left(\frac{1}{r}-\frac{1}{r-1}\frac{F(S^{(1)}_j\alpha_{r-1}\mid \hat{\sigma}^*)}{F(S^{(1)}_j\alpha_r\mid \hat{\sigma}^*)}\right)\right\}\leq 0.
\]

Hence, 
\[
\mathrm{FDR}\leq \sum_{j\in\mathcal{N}} E_{S^{(1)}}[E_{\hat{\sigma}}\{F(S^{(1)}_j\alpha/m\mid\hat{\sigma})\}] = \sum_{j\in\mathcal{N}} E_{S^{(1)}} [S^{(1)}_j\alpha/m]=(m_0/m)\alpha\leq \alpha. 
\]

\medskip
FDR control of \textit{Method~2}: 
For finding the FDR, we condition on $\hat{\sigma}$ and $S^{(1)}$ as in the previous proof. Following the steps of the proof of Theorem~1 of \citet{ramdas2019unified} and using Lemma~\ref{lem:asymptotic_pval}, we get
\[
E\left\{ \frac{\text{V}}{\text{R}} \middle| S^{(1)},\hat{\sigma} \right\} \leq \frac{\alpha}{m} \sum_{j \in \mathcal{N}} S^{(1)}_j E \left\{ \frac{1}{\hat{\pi}^{-j}_{0}} \middle| S^{(1)},\hat{\sigma} \right\} + o(1),
\]
where,
\[
\hat{\pi}^{-j}_{0} = \frac{1 + \sum_{k \neq j} \mathbb{I}(P^{(2)}_k > \lambda)}{m(1 - \lambda)}.
\]
Following \citet{benjamini2006adaptive} and using Lemma~\ref{lem:asymptotic_pval}, it is easy to get that 
\[
E \left\{ \hat{\pi}^{-j}_{0} \middle| S^{(1)},\hat{\sigma} \right\} \leq \frac{1}{\pi_0} + o(1).
\]
Thus,
\[
E\left\{ \frac{\text{V}}{\text{R}} \middle| S^{(1)},\hat{\sigma} \right\} \leq \frac{\alpha}{m\pi_0} \sum_{j \in \mathcal{N}} S^{(1)}_j + o(1),
\]
Note that \(E(S_j^{(1)}) \leq 1\) for \(j \in \mathcal{N}\), and hence by noting that the expression does not involve $\hat{\sigma}$ and taking the expectation over $S^{(1)}$, we get  
\[
\mathrm{FDR} \leq \frac{m_0\alpha}{m\pi_0} + o(1)=\alpha + o(1).
\]

\medskip
FDR control of \textit{Method~3}:  
Following the proof of Lemma~2.2 of \citet{biswas2023new} and using Lemma~\ref{lem:asymptotic_pval}, we can show that for $j\in\mathcal{N}$,
\[
E \left\{ \hat{\delta}^{-j}_{0} \middle| W^{(1)},\hat{\sigma} \right\} \leq \frac{m}{\max\{W^{(1)}_1, W^{(1)}_2, ..., W^{(1)}_m\}+\sum_{k\in\mathcal{N}\setminus\{j\}}W^{(1)}_k}+ o(1),
\]
where
\[
\hat{\delta}^{-j}_{0} = \frac{ \max\{W^{(1)}_1, W^{(1)}_2, ..., W^{(1)}_m\} + \sum_{k \neq j} W^{(1)}_k \mathbb{I}(P^{(2)}_k > \lambda)}{m(1 - \lambda)}.
\]
It is to be mentioned here that Lemma~3 of \citet{ramdas2019unified} plays an important role in Lemma~2.2 of \citet{biswas2023new} and the above inequality can also be directly derived from Lemma~3 of \citet{ramdas2019unified}. Following the proof of Theorem~3.1 of \citet{biswas2023new} and using the above inequality, it can be shown that
\[
E\left\{ \frac{\text{V}}{\text{R}} \middle| W^{(1)},\hat{\sigma} \right\} \leq \alpha\sum_{j\in\mathcal{N}} \frac{W^{(1)}_j}{\max\{W^{(1)}_1, W^{(1)}_2, ..., W^{(1)}_m\} + \sum_{k\in\mathcal{N}\setminus\{j\}}W^{(1)}_k}+ o(1).
\]
It is easy to check that the multiplier of $\alpha$ in the above inequality is less than or equal to $1$. Hence, 
\[
\mathrm{FDR} \leq \alpha + o(1).
\]

\end{proof}

\newpage

\section*{Appendix B}

\begin{figure}[h!]  
    \centering
    \includegraphics[width=\textwidth]{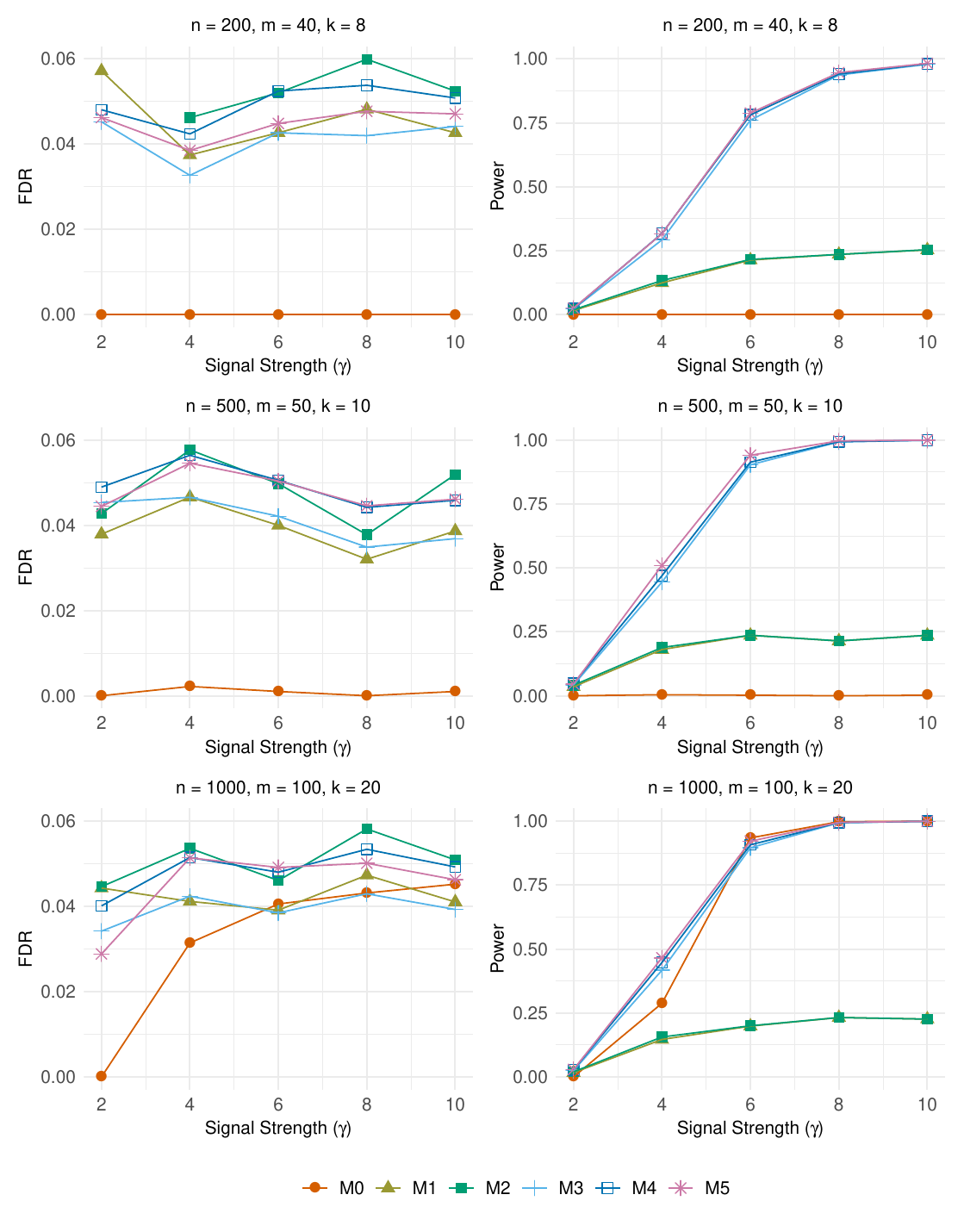}
    \caption{
        Comparison of false discovery rate and power across different settings $(n, m, k)$ with $\rho=0.1$ and $\alpha = 0.05$. Each row corresponds to a different simulation setting, and signal strength values ($\gamma = 2, 4, 6, 8, 10$) are shown on the x-axis.
    }
    
\end{figure}

\newpage

\begin{figure}[h!]  
    \centering
    \includegraphics[width=\textwidth]{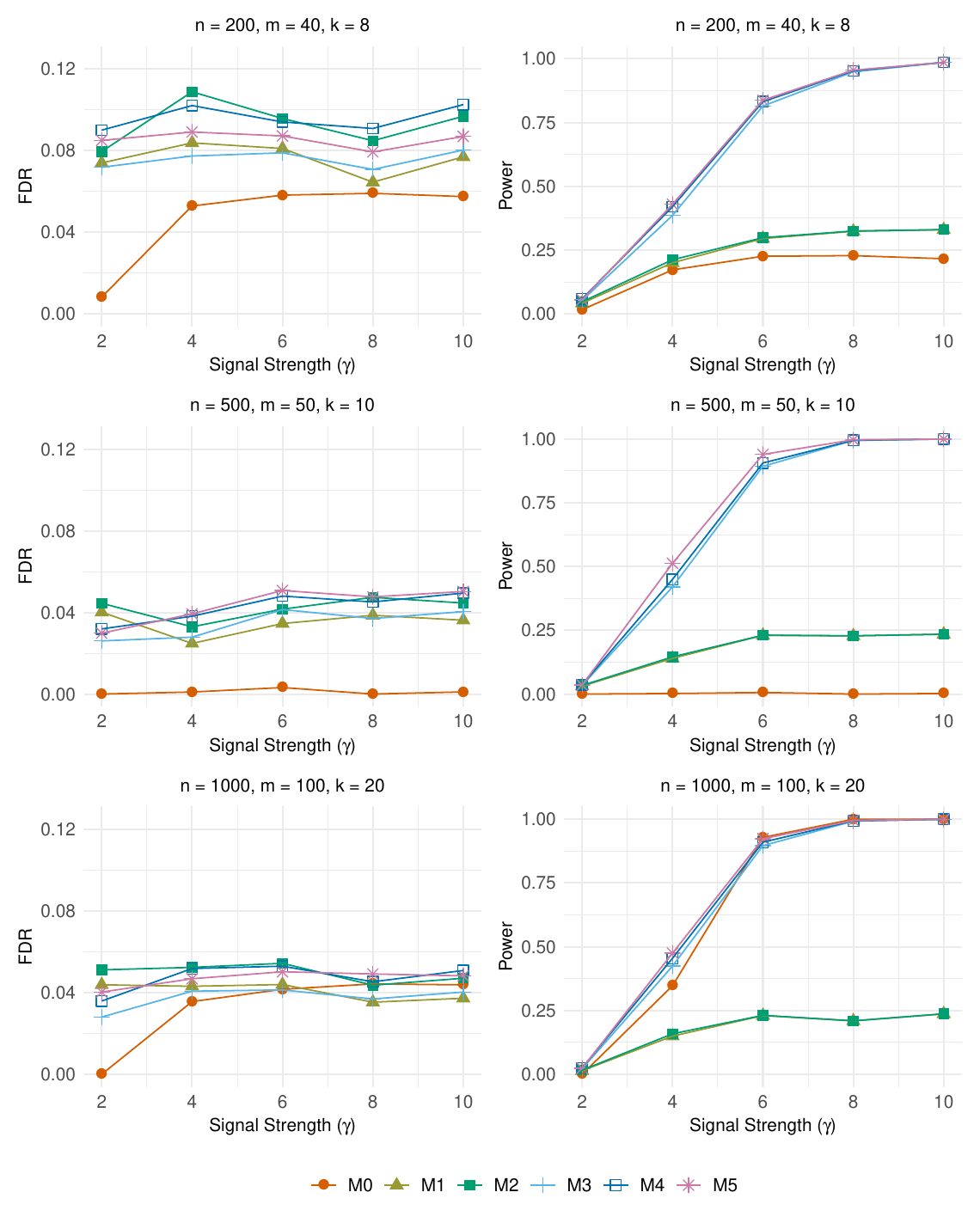}
    \caption{
        Comparison of false discovery rate and power across different settings $(n, m, k)$ with $\rho=0.1$ and $\alpha = 0.1$. Each row corresponds to a different simulation setting, and signal strength values ($\gamma = 2, 4, 6, 8, 10$) are shown on the x-axis.
    }
    
\end{figure}

\newpage

\begin{figure}[h!]  
    \centering
    \includegraphics[width=\textwidth]{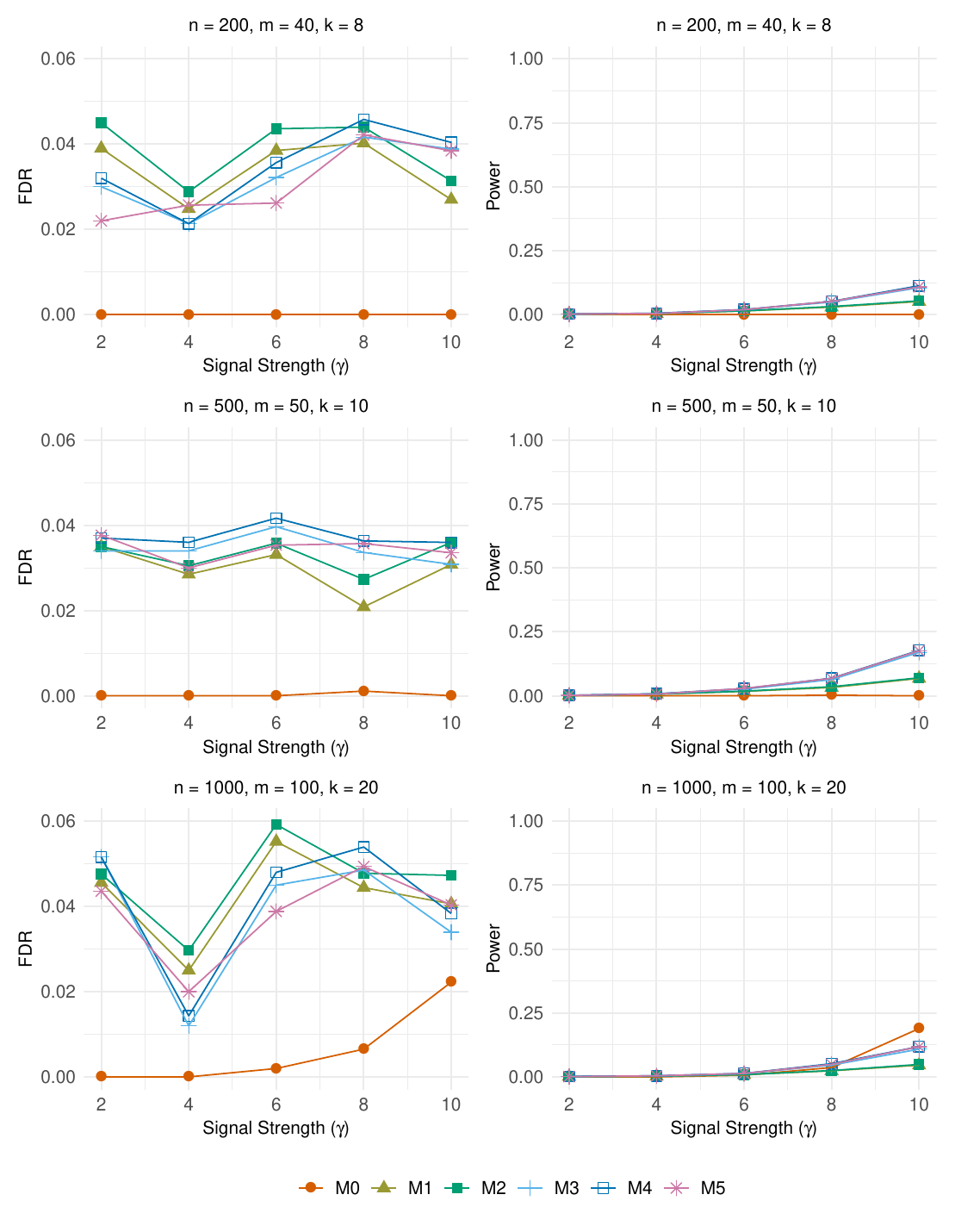}
    \caption{
        Comparison of false discovery rate and power across different settings $(n, m, k)$ with $\rho=0.9$ and $\alpha = 0.05$. Each row corresponds to a different simulation setting, and signal strength values ($\gamma = 2, 4, 6, 8, 10$) are shown on the x-axis.
    }
    
\end{figure}

\newpage

\begin{figure}[h!]  
    \centering
    \includegraphics[width=\textwidth]{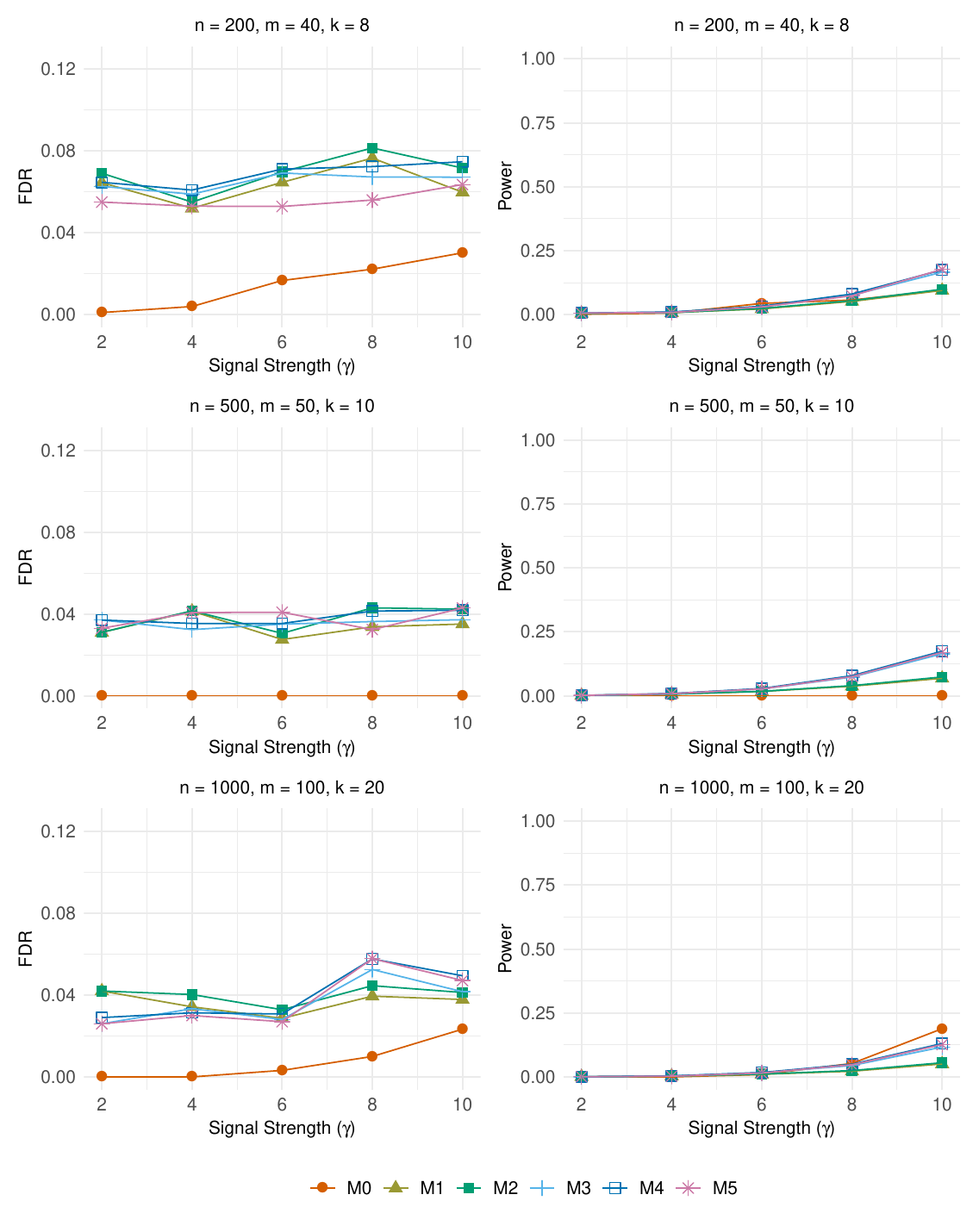}
    \caption{
        Comparison of false discovery rate and power across different settings $(n, m, k)$ with $\rho=0.9$ and $\alpha = 0.1$. Each row corresponds to a different simulation setting, and signal strength values ($\gamma = 2, 4, 6, 8, 10$) are shown on the x-axis.
    }
    
\end{figure}

\newpage

\begin{figure}[h!]  
    \centering
    \includegraphics[width=\textwidth]{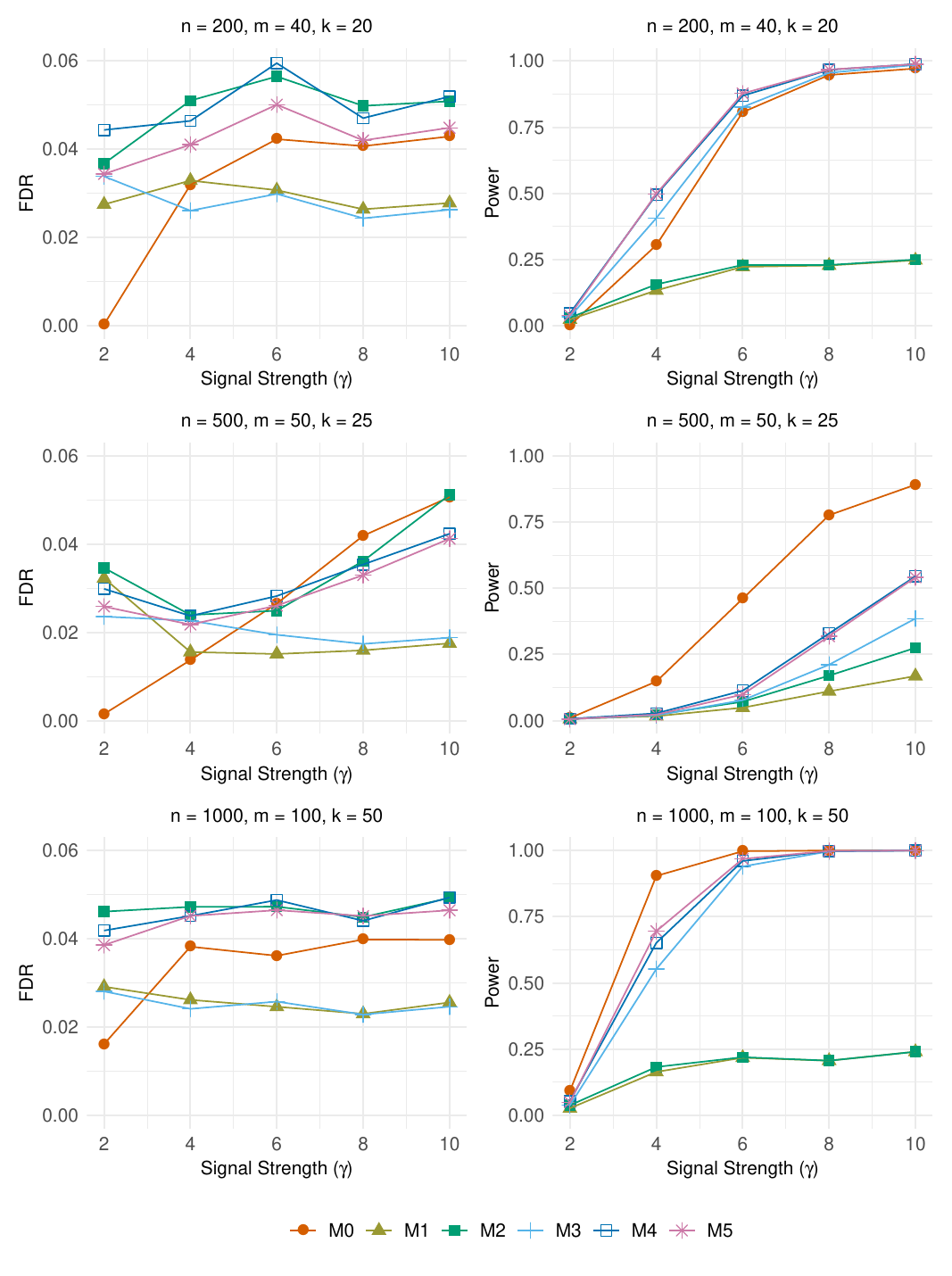}
    \caption{
        Comparison of false discovery rate and power across different settings $(n, m, k)$ with $\rho=0.1$ and $\alpha = 0.05$. Each row corresponds to a different simulation setting, and signal strength values ($\gamma = 2, 4, 6, 8, 10$) are shown on the x-axis.
    }
    
\end{figure}

\newpage

\begin{figure}[h!]  
    \centering
    \includegraphics[width=\textwidth]{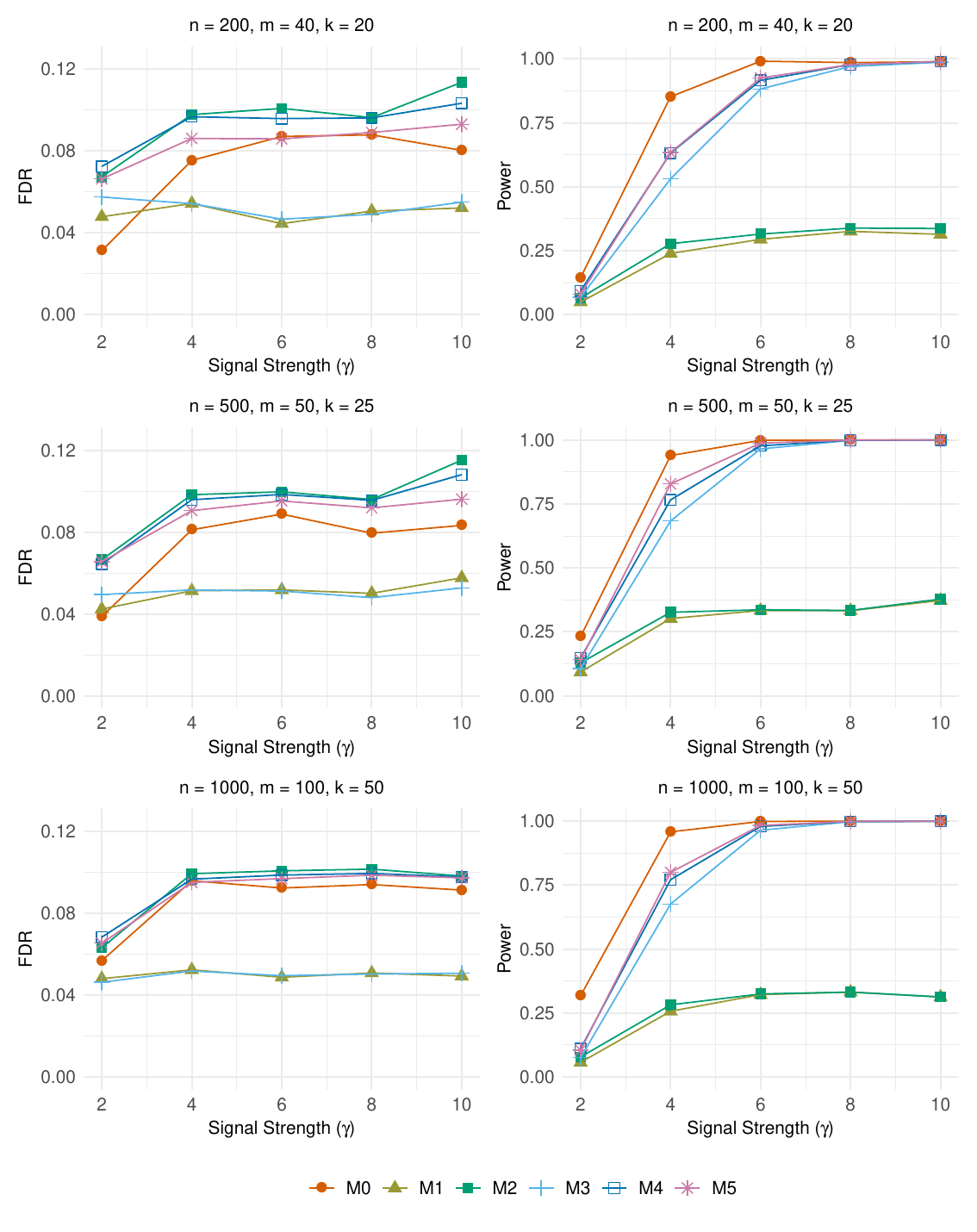}
    \caption{
        Comparison of false discovery rate and power across different settings $(n, m, k)$ with $\rho=0.1$ and $\alpha = 0.1$. Each row corresponds to a different simulation setting, and signal strength values ($\gamma = 2, 4, 6, 8, 10$) are shown on the x-axis.
    }
    
\end{figure}

\newpage

\begin{figure}[h!]  
    \centering
    \includegraphics[width=\textwidth]{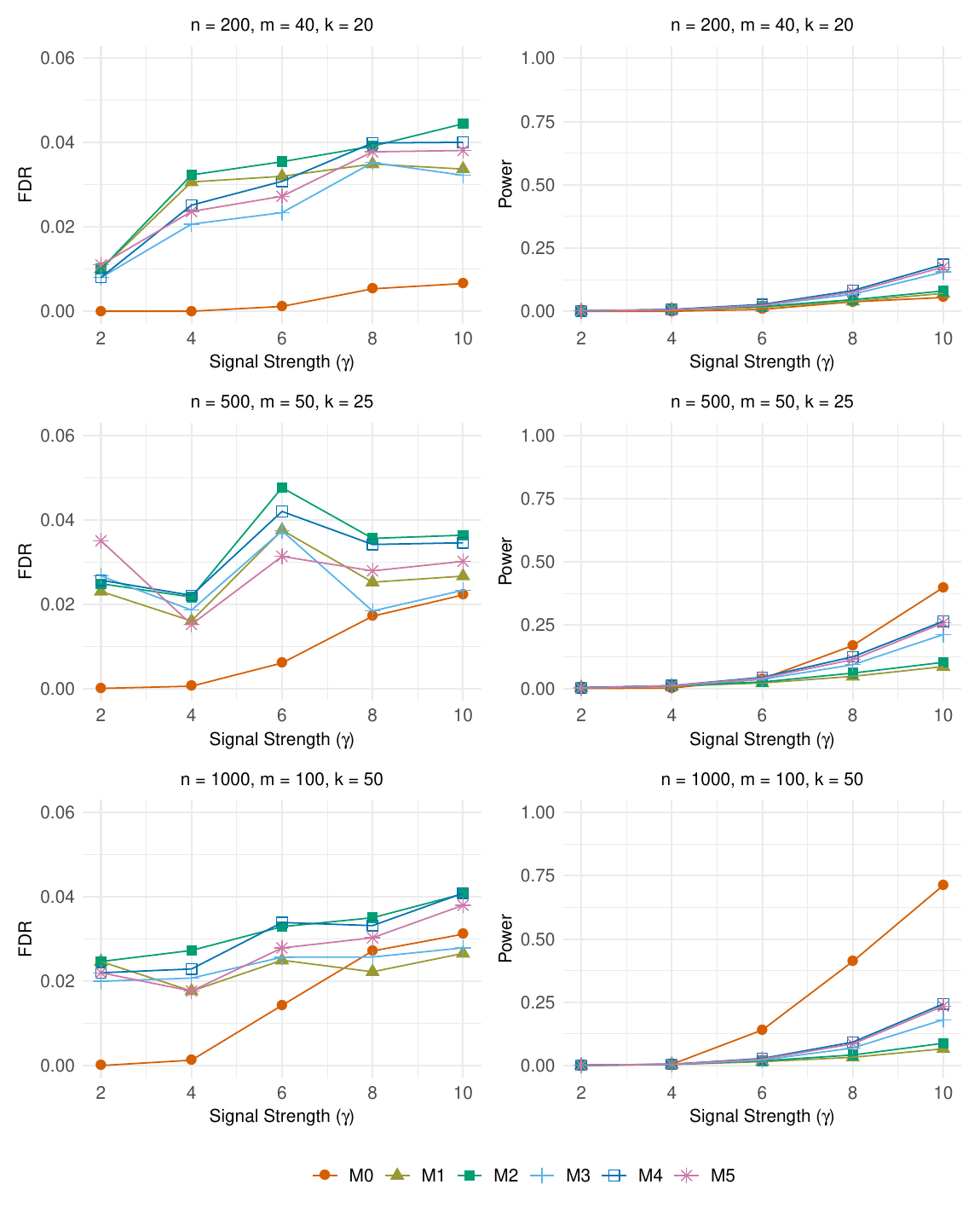}
    \caption{
        Comparison of false discovery rate and power across different settings $(n, m, k)$ with $\rho=0.9$ and $\alpha = 0.05$. Each row corresponds to a different simulation setting, and signal strength values ($\gamma = 2, 4, 6, 8, 10$) are shown on the x-axis.
    }
    
\end{figure}

\newpage

\begin{figure}[h!]  
    \centering
    \includegraphics[width=\textwidth]{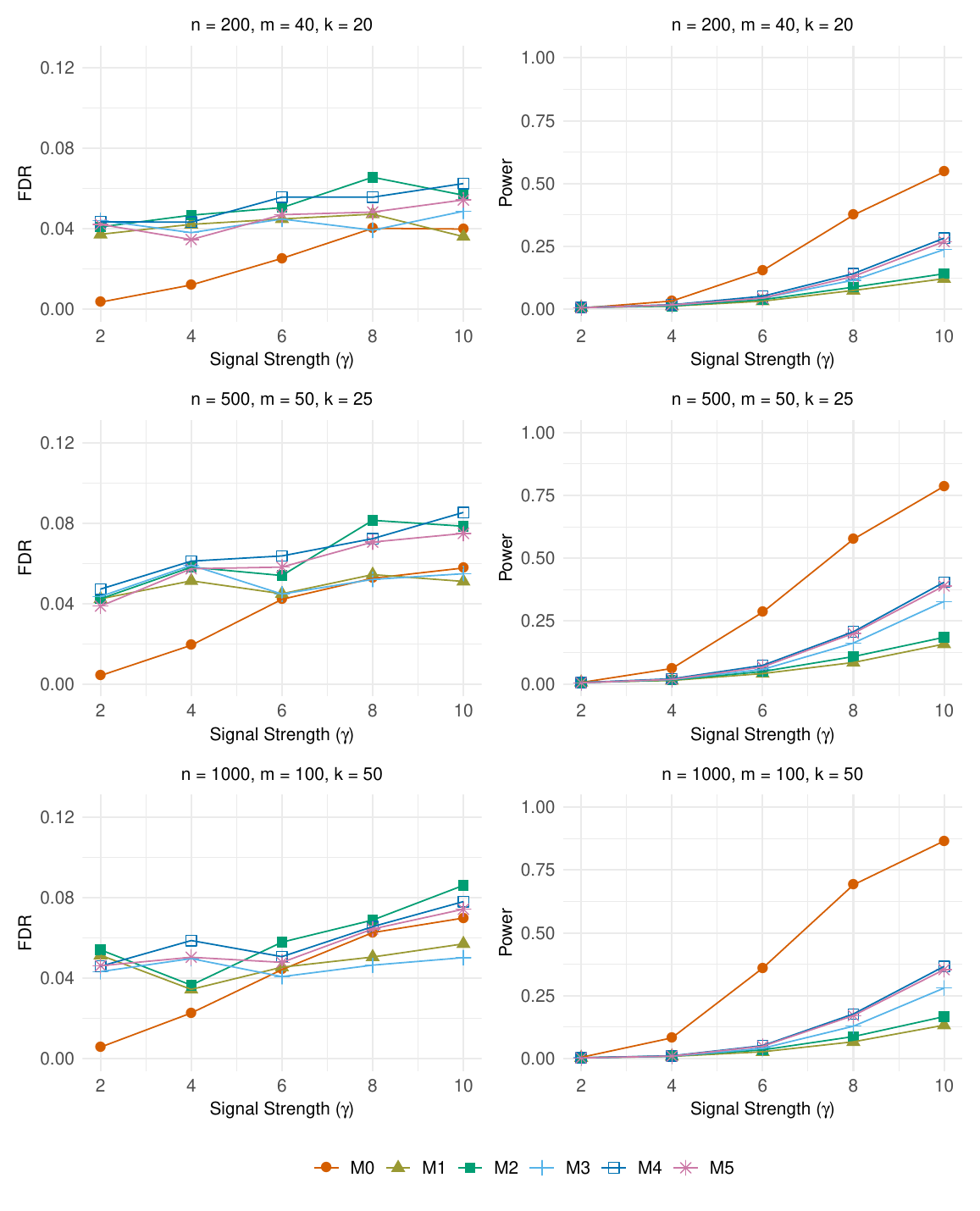}
    \caption{
        Comparison of false discovery rate and power across different settings $(n, m, k)$ with $\rho=0.9$ and $\alpha = 0.1$. Each row corresponds to a different simulation setting, and signal strength values ($\gamma = 2, 4, 6, 8, 10$) are shown on the x-axis.
    }
    
\end{figure}

\end{document}